\documentclass[12pt]{article}
\usepackage{amsmath}
\usepackage{graphicx}
\usepackage{enumerate}
\usepackage{natbib}
\usepackage{url} 
\usepackage{authblk}
\newcommand{\blind}{1}

\allowdisplaybreaks[4]
\usepackage{float}
\usepackage{mathrsfs}
\usepackage{amsfonts}   
\usepackage{scalerel}
\usepackage{times}
\usepackage{bm}
\usepackage{colonequals}
\usepackage[plain,noend]{algorithm2e}
\usepackage{caption}
\usepackage{cmll}
\usepackage{geometry}
\usepackage{setspace}
\usepackage{amsthm}
\usepackage{color}
\usepackage{multirow}
\usepackage{booktabs}
\usepackage{subcaption}
\makeatletter

\newtheorem{theorem}{\bf{Theorem}}
\newtheorem{remark}{\bf{Remark}}
\newtheorem{lemma}{\bf{Lemma}}

\newtheorem{example}{\bf{Example}}
\newtheorem{proposition}{\bf{Proposition}}
\newtheorem{condition}{\bf{Condition}}

\newcommand{\Var}{\rm Var}
\newcommand{\var}{{\rm var}}
\newcommand{\cF}{\mathcal{F}}
\newcommand{\cT}{\mathcal{T}}
\newcommand{\cL}{\mathcal{L}}
\newcommand{\cW}{\mathcal{W}}
\newcommand{\cP}{\mathcal{P}}
\newcommand{\bV}{V}
\newcommand{\bU}{U}
\newcommand{\bX}{X}
\newcommand{\bZ}{Z}

\newcommand{\bzeta}{\zeta}

\newcommand{\cS}{\scaleto{\mathcal{S}}{5pt}}
\newcommand{\cC}{\scaleto{\mathcal{C}}{5pt}}
\newcommand{\cG}{\scaleto{\mathcal{G}}{5pt}}
\makeatother

\graphicspath{{./figures/}}

\def\T{{ \mathrm{\scriptscriptstyle T} }}

\begin{document}

		\def\spacingset#1{\renewcommand{\baselinestretch}%
			{#1}\small\normalsize} \spacingset{1}

		
		\if1\blind
		{
			\title{\bf A Maximin Optimal Approach for Sampling Designs in Two-phase Studies}
			\author[a,b]{Ruoyu Wang}
			\author[a]{Qihua Wang\thanks{Email: qhwang@amss.ac.cn}}
			\author[c]{Wang Miao}
			\affil[a]{Academy of Mathematics and Systems Science, Chinese Academy of Sciences, Beijing, P.R. China}
			\affil[b]{Department of Biostatistics, Harvard University, Boston, USA}
			\affil[c]{Department of Probability and Statistics, Peking University, Beijing, P.R. China}
			\date{}
			\maketitle
		} \fi
		
		\if0\blind
		{
			\bigskip
			\bigskip
			\bigskip
			\begin{center}
				{\LARGE\bf A Maximin Optimal Approach for Sampling Designs in Two-phase Studies}
			\end{center}
			\medskip
		} \fi
		
		\bigskip
		\begin{abstract}
		Data collection costs can vary widely across variables in data science tasks. Two-phase designs can be employed to save data collection costs. This paper considers the two-phase studies where inexpensive variables are collected for all subjects in the first phase, and expensive variables are measured for a subsample of subjects in the second phase based on a predetermined sampling rule. The estimation efficiency under two-phase designs relies heavily on the sampling rule. Existing literature primarily focuses on designing sampling rules for estimating a scalar parameter in some parametric models or specific estimating problems. However, real-world scenarios are usually model-unknown and involve two-phase designs for model-free estimation of a scalar or multi-dimensional parameter. This paper proposes a maximin criterion to design an optimal sampling rule based on semiparametric efficiency bounds. The proposed method is model-free and applicable to general estimating problems. The resulting sampling rule can minimize the semiparametric efficiency bound when the parameter is scalar and improve the bound for every component when the parameter is multi-dimensional. Simulation studies demonstrate that the proposed designs reduce the variance of the resulting estimator in various settings. The implementation of the proposed design is illustrated in a real data analysis.
		\end{abstract}
		
		\noindent%
		\vfill
		
		\newpage
		\spacingset{1.5} 
	
	\section{Introduction} 
	The research of modern data science often involves variables with different collection costs. For example, the research objective in medical and genetic studies often involves analyzing the associations between disease status and biomarkers while adjusting for demographic factors such as age, gender, or race. 
	Measuring the biomarker can be fairly expensive compared to disease status and demographic factors that can be obtained through questionnaires or electronic health records. The two-phase design is commonly employed to achieve cost-effectiveness in ecological, epidemiological and genetic researches \citep{fattorini2017design, nab2021sampling,lin2013quantitative}, especially in large-scale studies such as the National Wilms' Tumor Study \citep{green1998comparison}, the Caribbean, Central, and South America network for HIV Epidemiology \citep{mcgowan2007cohort}, and the National Heart, Lung, and Blood Institute Exome Sequencing Project \citep{lin2013quantitative}.
	
	In two-phase studies considered in e.g., \cite{mcnamee2002optimal,mcisaac2014response,mcisaac2015adaptive} and \cite{tao2020optimal}, inexpensive variables are collected for all subjects in the first phase, while expensive variables are measured only on a subsample drawn according to some sampling rule in the second phase. The parameter of interest is estimated using observations from the two phases. 
	%
	The design of the second phase sampling rule is crucial to the efficiency of the resulting estimator. Many existing works develop sampling rules tailored for parameter estimation in conditional density models. \cite{zhou2014semiparametric} propose to fit a model to predict the expensive variables of different subjects and draw a sample from the subjects with extreme predicted values in the second phase. The proposed sampling rule performs well in the empirical studies in \citep{zhou2014semiparametric}. However, there is no theoretical guarantee for the optimality of the proposed sampling rules. Recently, researchers have developed optimal sampling rules that minimize the asymptotic variances of some specific estimators or semiparametric efficiency bound under the conditional density model. \cite{mcisaac2015adaptive} investigate the optimal sampling rule for the mean score estimator with discrete inexpensive variables. \cite{chen2022optimal} derive the optimal sampling rule for the inverse-probability weighted estimator and generalized raking estimator for parameter estimation in conditional density models. \cite{tao2020optimal} propose the optimal sampling rule that minimizes the semiparametric efficiency bound in a general
	conditional density estimation problem which may contain continuous inexpensive
	variables and nuisance parameters.
	A few existing works also consider other estimation problems than the conditional density model. \cite{mcnamee2002optimal} proposes the optimal rule for estimating sensitivity, specificity, and positive predictive value of a test with categorical test results and no covariate. \cite{gilbert2014optimal} study the optimal design that minimized the semiparametric efficiency bound for mean or mean difference estimation with discrete inexpensive variables and proposed to stratify the continuous inexpensive variables. \cite{zhang2024patient} explore a problem similar to \cite{gilbert2014optimal} while extending the scope by accounting for the selection bias in the first-phase data.
	
	Despite the fruitful works on sampling rule designs in the second phase, there still exist several challenges in practice. Firstly, current research on optimal design mainly focuses on parameter estimation in certain specific models, leaving many important problems unconsidered. Secondly, in many applications, there is a multi-dimensional parameter of interest; however, existing optimal designs work best for scalar parameters, and optimal designs for a multi-dimensional parameter vector are complicated and remain unstudied in the literature. The optimal sampling rule for one component does not necessarily lead to good estimation efficiency for other components. For instance, this issue arises when one estimate sensitivity, specificity, and positive predictive value simultaneously as illustrated by \citep{mcnamee2002optimal}. Therefore, it is crucial to strike a balance between estimation efficiency for different components to achieve efficiency gains for every component. 
	
	This paper proposes a model-free design approach for estimating a generic population parameter of interest, which can be either scalar or multi-dimensional. The proposed method applies to many important but understudied problems in two-phase sampling design, such as the estimation of quantiles, Pearson correlation coefficients, and average treatment effects. When designing the sampling rule, the proposed method focuses on optimizing the resulting semiparametric efficiency bound, which represents the smallest asymptotic variance that one can achieve for estimating the parameter of interest under regularity conditions. Once a sampling rule is established, various existing semiparametric methods can be employed to construct an efficient estimator whose asymptotic variance achieves the semiparametric efficiency bound. These methods include the targeted maximum likelihood method \citep{van2006targeted}, the estimating equation method \citep{tsiatis2007semiparametric}, and the one-step estimation method \citep{bickel1982adaptive}. Compared to existing works, our main contributions are as follows:
	\begin{itemize}
		\item  When the parameter is scalar, we derive the optimal sampling rule that minimizes the semiparametric efficiency bound subject to a budget constraint. In contrast to many existing methods, our approach does not make any parametric model assumptions and allows the	inexpensive variables to be continuous. Obtaining the optimal sampling rule is challenging	in this case because it involves an infinite-dimensional constrained optimization problem.
		\item We study the problem of two-phase sampling design with a multi-dimensional parameter of interest, which
		remains largely unexplored in the literature to our knowledge. The problem becomes more challenging when the parameter is multi-dimensional. Specifically, under the optimal sampling rule for one given component, the semiparametric efficiency bounds of other components may be larger than those under the uniform sampling rule that includes each subject with a constant probability. To resolve this issue, we define an objective function based on the semiparametric efficiency bound for each component and design the sampling rule by maximizing the objective function, which leads to a \emph{maximin optimal sampling rule}. The maximin optimal sampling rule outperforms the uniform rule regarding the semiparametric efficiency bound for every component under mild conditions. Nevertheless, it involves an intractable infinite-dimensional constrained maximin problem to obtain the proposed sampling rule. To tackle this, we design a novel finite-dimensional optimization problem that shares the same solution as the original infinite-dimensional problem, and obtain the solution by routine optimization algorithms.  
	\end{itemize}
	
	The proposed sampling rules depend on unknown quantities determined by the data generating process. 
	We adopt the standard ``pilot sample" approach to obtain the proposed sampling rules \citep{gilbert2014optimal,tao2020optimal}. To be specific, we select a simple random sample, referred to as a ``pilot sample", at the beginning of the second phase and then use 
	the pilot sample to estimate the unknown quantities. The remaining subjects are then selected based on the estimated sampling rule. 
	We establish the optimality of the proposed sampling rules and the consistency of their estimators. Simulation studies demonstrate that the proposed designs can achieve substantial efficiency gains compared to the uniform rule under various settings. The implementation of the proposed design is further illustrated in a real data example. 
	
	The proposed sampling rules are derived based on the efficient influence function of an estimation problem. The method to derive the efficient influence function is standard and thoroughly investigated in literature of semiparametric statistics. See, e.g., \cite{bickel1982adaptive,tsiatis2007semiparametric, van2012unified}. This implies that the proposed method is versatile and can be applied to a wide range of estimation problems as long as the corresponding efficient influence function is available. For instance, it can be employed in two-phase designs for data validation in electronic health records to address issues related to missing data and measurement errors \citep{lotspeich2022efficient, zhang2024patient}. Additionally, this method offers optimal designs for important but underexplored tasks, such as the selection of validation data in observational studies with missing confounders \citep{lin2014adjustment, yang2019combining}, and quantitative trait analysis with multiple traits of interest \citep{lin2013quantitative}.
	
	
	\section{Two-phase design and efficiency bound}\label{sec: efficiency bound}
	Let $U$ denote a vector of expensive variables, such as true disease status or biomarkers, and let $V$ denote a vector of inexpensive variables, such as age, gender, race, or surrogate measures for $U$.
	In a two-phase study, $\bV$  is collected in the first phase for all $n$ subjects in the study. In the second phase, a subsample of subjects is drawn from the $n$ subjects according to a sampling rule that may depend on $\bV$, and $\bU$ is measured for the subsample. Let $R$ be the sampling indicator in the second phase. $R = 1$ if the subject is included in the second stage and $\bU$ is observed; $R = 0$, otherwise. The inclusion probability for the second stage depends on the first-phase variable vector $\bV$, that is, $P(R = 1\mid \bV) = \rho(\bV)$ where $\rho(\cdot)$ is the sampling rule. We refer to $\bV$ and $\bU$ as the first-phase and second-phase variables, respectively. The sampling framework considered here is commonly adopted in the literature of two-phase studies \citep{chatterjee2003pseudoscore,gilbert2014optimal,tao2017efficient,tao2020optimal}, and has been widely used in epidemiological and genetic studies to investigate the effect of various risk factors and biomarkers on different diseases-related traits \citep{nab2021sampling, lin2013quantitative}.
	
	Suppose $\theta_{0}$ is the parameter of interest in the two-phase study, which is possibly a  multi-dimensional functional of the joint distribution of $(\bV,\bU)$. An important question is how to design the sampling rule $\rho(\cdot)$ in the second phase to minimize the asymptotic variance for estimating $\theta_{0}$. The optimal sampling rule that minimizes the asymptotic variance is usually model-dependent and estimator-specific \citep{chen2022optimal}. We here consider the model-free semiparametric efficient estimating problem and develop an optimal sample rule by minimizing the semiparametric efficient bound for estimating $\theta_{0}$. The semiparametric efficiency bound \citep{bickel1982adaptive,tsiatis2007semiparametric} is an extension of the Cramer-Rao bound to the semiparametric and nonparametric setting, which characterizes the smallest asymptotic variance one can achieve for estimating $\theta_{0}$ under the considered distribution class.
	Besides, researchers usually intend to control the sampling fraction $E[\rho(\bV)]$ in the second phase to manage the budget since $\bU$ is expensive to measure. This further motivates us to search for the sampling rule that minimizes the semiparametric efficiency bound under the constraint that $E[\rho(\bV)]$ is no larger than a given threshold.
	
	Next, we introduce the efficient influence function and the semiparametric efficiency bound under the two-phase design. The efficient influence function is an important concept in semiparametric theory which is useful in deriving the semiparametric efficiency bound and constructing efficient estimators that attain the bound \citep{tsiatis2007semiparametric}.
	Suppose $\psi(\bV,\bU)$ is the efficient influence function for estimating $\theta_{0}$ under a full data setting where $\bV$ and $\bU$ are observed for every subject. We write $\psi(\bV,\bU)$ as $\psi$ for short, wherever it does not cause confusion. Under two-phase designs, both the efficient influence function and the semiparametric efficiency bound are highly dependent on $\psi$. The expression of $\psi$ is well established in many important problems across survey sampling and modern epidemiological and clinical studies, including the outcome mean estimation, linear regression analysis, and average treatment effect estimation \citep{tsiatis2007semiparametric}. See Section \ref{sec: examples of EIF} in Supplementary Material for the expression of $\psi$ in the exemplified problems.
	
	Based on $\psi$, we can derive the efficient influence function and the semiparametric efficiency bound under the two-phase design without parametric model assumptions. 
	Assume throughout this paper that the distribution of $(\bV, \bU)$ has a density with respect to some dominance measure, e.g., the counting measure or the Lebesgue measure. We can derive semiparametric efficiency bound under the two-phase design in the following lemma by using techniques similar to those in \cite{tsiatis2007semiparametric} for missing data problems. 
	\begin{lemma}\label{lem: EIF}
		Let $\psi$ be the full data efficient influence function for estimating $\theta_{0}$ and $\rho(\cdot)$  the sampling rule in the second phase. Under the regularity conditions in Section \ref{subsec: model settings} in Supplementary Material,
		the efficient influence function under the two-phase design is 
		\begin{equation}\label{eq: EIF}
			\frac{R\psi}{\rho(\bV)} -\left(\frac{R}{\rho(\bV)} - 1\right)\Pi(\bV),
		\end{equation}
		and the semiparametric efficiency bound is
		\begin{equation*}
			E\left[\frac{\Sigma(\bV)}{\rho(\bV)}\right] + \Var[\Pi(\bV)],
		\end{equation*}
		where $\Pi(\bV) = E[\psi\mid \bV]$ and $\Sigma(\bV) = \Var[\psi\mid \bV]$ are the mean and variance of $\psi$ conditional on $\bV$, respectively.
	\end{lemma}
	
	Several methods can be adopted to obtain estimators for $\theta_{0}$ that have the influence function \eqref{eq: EIF} and hence achieve the semiparametric efficiency bound. These methods include the targeted maximum likelihood method \citep{van2006targeted}, the estimating equation method \citep{tsiatis2007semiparametric}, and the one-step estimation method \citep{bickel1982adaptive}. Note that the semiparametric efficiency bound depends on $\rho(\cdot)$. Therefore, the efficiency of the efficient estimator can be further improved by employing a carefully designed sampling rule. 
	We delve into this topic in the subsequent sections.
	
	\section{Optimal design in two-phase studies}\label{sec: opt design}
	\subsection{Optimal design for a scalar parameter}\label{sec: scalar}
	In this section, we consider the optimal sampling rule when $\theta_{0}$ is scalar. Denote the  variance of $\psi$ conditional on $V$ as $\sigma^{2}(\bV)$ and assume $\sigma^{2}(V) > 0$. The semiparametric efficiency bound for  $\theta_{0}$ is
	\begin{equation}\label{eq: bound scalar}
		E\left[\frac{\sigma^{2}(\bV)}{\rho(\bV)}\right] + \Var[\Pi(\bV)].
	\end{equation}
	We search for the optimal sampling rule that minimizes the semiparametric efficiency bound \eqref{eq: bound scalar} subject to the constraint 
	\begin{equation}\label{eq: budget constraint second phase only}
		E[\rho(\bV)] \leq \varpi,
	\end{equation} 
	where $\varpi \in (0,1]$ is the maximal sampling fraction determined by study budgets. The problem is an infinite-dimensional optimization problem that is generally difficult to solve. 
	The constraint that the sampling probability belongs to $[0,1]$  further complicates the problem. Nonetheless, we find that the optimal sampling rule can be obtained via some constructive arguments. For any positive function $\lambda(\cdot)$ and positive $\tau$, let 
	\[
	\rho(\bV; \lambda, \tau) = 1\{\lambda(\bV) > \tau\} + \frac{1}{\tau}\lambda(\bV)1\{\lambda(\bV) \leq \tau\}.
	\]
	Note that $\rho(\bV;\sigma, \tau)$ is continuous and strictly decreasing with respect to $\tau$, where $\sigma(\cdot)$ is the conditional variance function of $\psi$ conditional on $V$. This implies that the equation $E[\rho(\bV; \sigma, \tau)] = \varpi$ has a unique solution, denoted by $\tau_{\cS}$.
	The form of the optimal sampling rule is obtained in the following theorem.
	\begin{theorem}\label{thm: optimal probability}
		The optimal sampling rule that minimizes the semiparametric efficiency bound \eqref{eq: bound scalar} under the constraint \eqref{eq: budget constraint second phase only} is   
		\[
		\rho_{\cS}(\cdot) = \rho(\cdot; \sigma, {\tau_{\cS}}) = 1\{\sigma(\cdot) > \tau_{\cS}\} + \frac{1}{\tau_{\cS}}\sigma(\cdot)1\{\sigma(\cdot) \leq \tau_{\cS}\},
		\]
		where $\tau_{\cS}$ is the unique solution of $E[\rho(\bV; \sigma, \tau)] = \varpi$ with respect to $\tau$.
	\end{theorem}
	Please refer to Section \ref{subsec: proof opt prob scalar} in Supplementary Material for the proof of Theorem \ref{thm: optimal probability}. Let us provide some intuitions for Theorem \ref{thm: optimal probability}. Note that minimizing the semiparametric efficiency bound \eqref{eq: bound scalar} is equivalent to minimizing $E[\sigma^{2}(\bV)/\rho(\bV)]$. Notice that $E[\rho(\bV)] \leq \varpi$. We have 
	$\varpi E[\sigma^{2}(\bV)/\rho(\bV)]\geq E[\rho(\bV)]E[\sigma^{2}(\bV)/\rho(\bV)] \geq (E[\sigma(\bV)])^{2}$ according to the Cauchy--Schwarz inequality and the equality hold only if $\rho(\bV) = \varpi\sigma(\bV) / E[\sigma(\bV)]$. Let $\tau^{*} = E[\sigma(\bV)] / \varpi$. Then one may conjecture that $\sigma(\bV)/ \tau^{*}$ is the optimal sampling rule that minimizes the semiparametric efficiency bound. However, $\sigma(\bV)/ \tau^{*}$ may not be a feasible sampling rule as it can be larger than one. Thus we consider the truncated version $\rho(\bV; \sigma, \tau) = 1\{\sigma(\bV) > \tau\} + \sigma(\bV)1\{\sigma(\bV) \leq \tau\}/\tau = \min\{\sigma(\bV)/\tau,1\}$ and show that it is the desired optimal rule if we determine $\tau$ by solving the equation $E[\rho(\bV; \sigma, \tau)] = \varpi$.
	
	Theorem \ref{thm: optimal probability} gives the optimal sampling rule $\rho_{\cS}(\cdot)$ for estimating a general population parameter.
	The optimal sampling rule $\rho_{\cS}(\cdot)$ is determined by $\sigma(\cdot) = \var(\psi\mid V = \cdot)$. Thus, to estimate the optimal sampling rule, one only needs to estimate the conditional variance function, which boils down to a well-studied problem in statistics.
	The form of $\rho_{\cS}(\cdot)$ is analogous to the Neyman allocation \citep{cochran2007sampling} for mean estimation problems in survey sampling which has been extended to scalar mean and regression coefficient estimation problems in two-phase studies \citep{reilly1995mean,gilbert2014optimal,mcisaac2014response,chen2022optimal}. Our result further extends the Neyman allocation to the problem of estimating a generic parameter. In addition, existing results often assume that the first-phase variables $V$ are discrete and propose to discretize the continuous variables regardless of the information loss caused by discretization \citep{mcisaac2014response}. When $\bV$ is a discrete variable, minimizing the semiparametric efficiency bound under the budget constraint is a finite-dimensional optimization problem and can be solved by the Lagrange multiplier method.
	When $\bV$ is a continuous variable, it becomes more challenging to find the optimal sampling rule, as it requires optimization over an infinite-dimensional function space. We overcome this difficulty by constructively suggesting the sampling rule $\rho_{\cS}(\cdot)$ based on the ideas discussed in the last paragraph. 
	
	It should be pointed out that \cite{tao2020optimal} also explore the minimization problem of semiparametric efficiency bound with budget constraint. They focus on the parameter estimation problem under a semiparametric conditional density model. Their approach allows the first-phase variable to be continuous and simplifies the infinite-dimensional optimization problem by invoking the Neyman--Pearson lemma. However, the simplified problem generally remains infinite-dimensional and intractable, except under the logistic model or the linear model with Gaussian error. In contrast, this paper considers the estimation of a generic population parameter. Building on the constructive ideas outlined in previous paragraphs, Theorem \ref{thm: optimal probability} effectively converts the sampling rule design problem into a manageable conditional variance estimation problem and is applicable to a wide range of estimation tasks.

	\subsection{Maximin design for a multi-dimensional parameter}\label{subsec: multivariate}
	Theorem \ref{thm: optimal probability} obtains the optimal design for a scalar parameter. In practice, the parameter of interest $\theta_{0}$ may well be multi-dimensional, especially when analyzing multiple traits and biomarkers or evaluating multi-level treatments.
	In this section, we study the two-phase design for a multi-dimensional parameter.
	Let $\theta_{0} = (\theta_{01},\dots, \theta_{0d})^{\T}$ and $\psi = (\psi_{1}, \dots, \psi_{d})^{\T}$ be the $d$-dimensional parameter and the full data influence function, respectively. The semiparametric efficiency bound for the $j$th component $\theta_{0j}$ is 
	\[E\left[\frac{\sigma_{j}^{2}(\bV)}{\rho(\bV)}\right] + \Var[\Pi_{j}(\bV)],\]
	where $\sigma_{j}(\bV) = \sqrt{\Var[\psi_{j}\mid \bV]}$ and $\Pi_{j}(\bV)$ is the $j$th component of $\Pi(\bV)$ for $j = 1,\dots,d$. The optimal sampling rule that minimizes the semiparametric efficiency bound for $\theta_{0j}$ can be obtained according to Theorem \ref{thm: optimal probability}. Let us denote the optimal sampling rule for $\theta_{0j}$ by $\rho_{j}(\cdot)$. 
	Compared to the scalar-parameter case, the main difficulty for the multi-dimensional-parameter problems is that the optimal rules for different components may be different. Thus, there usually does not exist a sampling rule that  simultaneously minimizes the semiparametric efficiency bound for different parameter components.  
	
	One intuitive approach to determining a sampling rule for a multi-dimensional parameter is to minimize the sum of the semiparametric efficiency bounds for each component
	\begin{equation}\label{eq: sum cri}
		\sum_{j=1}^{d}E\left[\frac{\sigma_{j}^{2}(\bV)}{\rho(\bV)}\right] + \sum_{j=1}^{d}\Var[\Pi_{j}(\bV)].
	\end{equation}
	This criterion is analogous to the widely-used A-optimality in the experiment design literature \citep{kiefer1959optimum}.
	Arguments similar to those in the proof of Theorem \ref{thm: optimal probability} can show that the sampling rule minimizing \eqref{eq: sum cri} is 
	\[
	\begin{aligned}
		\rho_{\rm sum}(\cdot) & = \rho(\cdot; \sigma_{\rm sum}, \tau_{\rm sum}) \\
		& = 1\{\sigma_{\rm sum}(\cdot) >  \tau_{\rm sum}\} + \frac{1}{\tau_{\rm sum}}\sigma_{\rm sum}(\cdot)1\{\sigma_{\rm sum}(\cdot) \leq  \tau_{\rm sum}\},
	\end{aligned}
	\]
	where $\sigma_{\rm sum}(\cdot) = \sqrt{\sum_{j=1}^{d}\sigma_{j}^{2}(\cdot)}$ and $ \tau_{\rm sum}$ denotes the solution of the equation $E[\rho(\bV; \sigma_{\rm sum}, \tau)] = \varpi$. However, the sum of semiparametric efficiency bounds may not be suitable in practice when the interest lies on each component. For instance, it could be obscure to sum up the semiparametric efficiency bound for per capita income and sex ratio in a survey study.
	
	The most straightforward sampling rule that satisfies the budget constraint in Equation \eqref{eq: budget constraint second phase only} is the uniform sampling rule. Since the uniform rule is intuitive and easy to implement, a designed sampling rule is expected to perform better than the uniform rule. Unfortunately, neither $\rho_{\rm sum}(\cdot)$ nor $\rho_{j}(\cdot)$ ($j \in {1, \dots, d}$) fulfills this requirement.
	Specifically, the semiparametric efficiency bound for some component of $\theta_{0}$ may be larger under $\rho_{\rm sum}(\cdot)$ or $\rho_{j}(\cdot)$ than that under the uniform rule. A simple example illustrating this issue is provided in Section \ref{subsec: multidim}  in Supplementary Material.	Thus, it is important to design a sampling rule that performs better than the uniform rule. This problem, however, remains largely unexplored in the literature to our knowledge.
	
	In the same spirit, we consider a more general problem. Given a benchmark sampling rule $\rho_{0}(\cdot)$ that satisfies $E[\rho_{0}(\bV)] = \varpi$, such as the uniform rule, our objective is to search for a sampling rule under which the semiparametric efficiency bound for each parameter is smaller or at least no larger than the corresponding semiparametric efficiency bound under $\rho_{0}(\cdot)$. Hereafter, a sampling rule $\rho(\cdot)$ is considered to \emph{dominates} $\rho_{0}(\cdot)$ if the semiparametric efficiency bound for each parameter component under $\rho(\cdot)$ is smaller than that under $\rho_{0}(\cdot)$.
	
	Let $\xi_{j} = E[\sigma_{j}^{2}(\bV) / \rho_{0}(\bV)]$. Define $b_{j} = E\left[\sigma_{j}^{2}(\bV)/\rho_{0}(\bV)\right] + \Var[\Pi_{j}(\bV)]$ as the semiparametric efficiency bound for the $j$th component under the benchmark sampling rule $\rho_{0}(\cdot)$ for $j = 1,\dots, d$. For any $\rho(\cdot)$, the expression
	\[
	b_{j}^{-1}\left(\xi_{j} - E\left[\frac{\sigma_{j}^{2}(\bV)}{\rho(\bV)}\right]\right)
	\]
	measures the relative improvement in the semiparametric efficiency bound for the $j$th component compared to the benchmark rule $\rho_{0}(\cdot)$. We focus on relative improvements to provide a normalized measure that accounts for different scales in the efficiency bounds across components.
	Consider the minimum of the difference over all components 
	\[M(\rho) = \min_{j=1,\dots,d}\left\{b_{j}^{-1}\left(\xi_{j} - E\left[\frac{\sigma_{j}^{2}(\bV)}{\rho(\bV)}\right]\right)\right\}.\]
	When maximizing $M(\rho)$ over a candidate set of sampling rules $\cP$ that contains $\rho_{0}(\cdot)$,
	the resulting maximin optimal sampling rule $\rho_{*}(\cdot)$ satisfies $M(\rho_{*}) = \max_{\rho\in\cP}M(\rho) \geq M(\rho_{0}) = 0$. This ensures that the semiparametric efficiency bound for every component under $\rho_{*}(\cdot)$ is no larger than that under $\rho_{0}(\cdot)$.
	The idea to derive the sampling rule by solving the maximin problem $\max_{\rho\in\cP}M(\rho)$ resembles the minimax regret strategy used in decision theory \citep{cox1979theoretical}, prediction \citep{meinshausen2015maximin}, experimental design \citep{dette1997designing}, and hypothesis testing \citep{liu2022minimax} for various purposes. Here, the strategy is employed to ensure that the resulting sampling rule can improve the benchmark sampling rule in terms of the semiparametric efficiency bound for every parameter component, which is not precedented in the context of two-phase sampling design to our knowledge. 
	
	Although guaranteed to be no worse than $\rho_{0}(\cdot)$, the actual improvement that $\rho_{*}(\cdot)$ can achieve depends on the choice of $\cP$. For example, if $\cP$ is the singleton set $\{\rho_{0}(\cdot)\}$, then $\rho_{*}(\cdot) = \rho_{0}(\cdot)$ and no improvement can be achieved.
	On the other hand, if there exists a sampling rule $\rho(\cdot) \in \cP$ that dominates $\rho_{0}(\cdot)$, then we have $M(\rho_{*}) \geq M(\rho) > 0$ which implies $\rho_{*}(\cdot)$ also dominates $\rho_{0}(\cdot)$. Hence, a careful specification of the set $\cP$ is necessary to achieve significant efficiency improvement over $\rho_{0}(\cdot)$.
	
	We provide some sensible candidates for $\cP$. Recall that $\rho_{j}(\cdot)$ is the optimal sampling rule for $\theta_{0j}$ for $j=1,\dots,d$. A reasonable choice is to take $\cP = \cP_{\cC} \colonequals \{\rho_{0}(\cdot) + \sum_{j=1}^{d}w_{j}(\rho_{j}(\cdot) - \rho_{0}(\cdot))\}$ to be the set which consists of all convex combinations of the component-wise optimal sampling rule $\rho_{j}(\cdot)$ and the benchmark rule $\rho_{0}(\cdot)$. This results in the following constrained maximin problem
	\begin{equation}\label{eq: maximin constrained}
		\begin{aligned}
			&\max_{w\in \cW}\min_{j=1,\dots,d} \left\{b_{j}^{-1}\left(\xi_{j} -E\left[\frac{\sigma_{j}^{2}(\bV)}{\rho_{0}(\bV) + \sum_{k=1}^{d}w_{k}(\rho_{k}(\bV) - \rho_{0}(\bV))}\right]\right)\right\},
		\end{aligned}
	\end{equation}
	where $\cW=\{w=(w_{1}, \dots, w_{d}) :  \sum_{j=1}^{d}w_{j} \leq 1, \ 0\leq w_{j} \leq 1, \ \text{for} \ j=1,\dots,d\}$. In problem \eqref{eq: maximin constrained}, by taking the feasible set to be the set $\cP_{\cC}$, we ensure that the optimization problem is a finite-dimensional convex problem and hence allows for effective computation. Let $w_{\cC} = (w_{{\cC},1}, \dots, w_{{\cC},d})$ be the solution of the maximin problem \eqref{eq: maximin constrained}. Then we obtain the constrained maximin optimal sampling rule  
	\[\rho_{\cC}(\cdot) = \rho_{0}(\cdot) + \sum_{j=1}^{d}w_{{\cC},j}\{\rho_{j}(\cdot) - \rho_{0}(\cdot)\}\]
	that solves the constrained maximin problem \eqref{eq: maximin constrained}.
	By using the optimal sampling rule of each component as basis functions, this method takes the accuracy of each component into consideration. By taking a convex combination and solving the maximin problem \eqref{eq: maximin constrained}, we achieve a balance in the estimation efficiency of different components. The resulting sampling rule $\rho_{\cC}(\cdot)$ dominates $\rho_{0}(\cdot)$ as long as $(0,\dots, 0)$ is not a solution of problem \eqref{eq: maximin constrained}.  
	
	In problem \eqref{eq: maximin constrained}, the optimization problem is simplified by constraining the feasible set $\cP$ to be the finite-dimensional set $\cP_{\cC}$. However, the resulting sampling rule $\rho_{\cC}(\cdot)$ might be suboptimal because many promising sampling rules may not be in $\cP_{\cC}$.
	To address this concern, one may consider the problem of maximizing $M(\rho)$ over all sampling rules under the budget constraint, i.e., taking $\cP = \cP_{\cG} \colonequals \{\rho(\cdot): 0\leq \rho(\cdot)\leq 1, \ E[\rho(\bV)] \leq \varpi\}$. Then we arrive at the maximin problem
	\begin{equation}\label{eq: maximin}
		\max_{\rho \in \cP_{\cG}} \min_{j=1,\dots,d} \left\{b_{j}^{-1}\left(\xi_{j} - E\left[\frac{\sigma_{j}^{2}(\bV)}{\rho(\bV)}\right]\right)\right\}.
	\end{equation}
	Solving the problem \eqref{eq: maximin} directly is challenging due to the functional nature of $\rho(\cdot)$ and the fact that it involves an infinite-dimensional maximin problem.
	To resolve this issue, by invoking the result in Theorem \ref{thm: optimal probability}, we design a novel finite-dimensional optimization problem that shares the same solution as \eqref{eq: maximin} .    
	\begin{theorem}\label{thm: finite-dimensional equiv}
		The global maximin optimal sampling rule that solves problem \eqref{eq: maximin} is 
		\[
		\begin{aligned}
			\rho_{\cG}(\cdot) = \rho(\cdot; \sigma_{w_{\cG}}, \tau_{w_{\cG}})
			= 1\{\sigma_{w_{\cG}}(\cdot) >  \tau_{w_{\cG}}\} + \frac{1}{\tau_{w_{\cG}}}\sigma_{w_{\cG}}(\cdot)1\{\sigma_{w_{\cG}}(\cdot) \leq  \tau_{w_{\cG}}\},
		\end{aligned}
		\]
		where $w_{\cG} = (w_{{\cG}, 1}, \dots, w_{{\cG}, d})^{\T}$ is the solution of the problem
		\begin{equation}\label{eq: min-max}
			\min_{w\in \cW^{\dag}} \left\{\sum_{j=1}^{d}w_{j}b_{j}^{-1}\xi_{j} - E\left[\sigma_{w}(\bV)\max\{\sigma_{w}(\bV), \tau_{w} \}\right]\right\},
		\end{equation}
		$\cW^{\dag}=\{w=(w_{1}, \dots, w_{d}) :  \sum_{j=1}^{d}w_{j} = 1, \ 0\leq w_{j} \leq 1, \ \text{for} \ j=1,\dots,d\}$, $\sigma_{w}(\cdot) = \sqrt{\sum_{j=1}^{d}w_{j}b_{j}^{-1}\sigma_{j}^{2}(\cdot)}$
		and $\tau_{w}$ is the unique solution of $E[\rho(\bV;\sigma_{w},\tau)] = \varpi$ with respect to $\tau$.
	\end{theorem}
	The proof of Theorem \ref{thm: finite-dimensional equiv} can be found in Section \ref{subsec: proof finite equiv}  in Supplementary Material. Similar arguments as in the proof of Theorem \ref{thm: optimal probability} can show that $\rho_{\cG}$ actually minimizes a weighted sum of the semiparametric efficiency bound for parameter component
	\begin{equation}\label{eq: weighted efficiency bound}
		\sum_{j=1}^{d}w_{{\cG}, j}b_{j}^{-1}\left\{E\left[\frac{\sigma_{j}^{2}(\bV)}{\rho^(\bV)}\right] + \Var[\Pi_{j}(\bV)]\right\}.
	\end{equation}
	From the proof of Theorem \ref{thm: finite-dimensional equiv}, it can be seen that the objective function in \eqref{eq: min-max} is actually a dual problem of \eqref{eq: maximin}, which can be viewed as a criterion to adaptively determine the weight of the efficiency bound for each parameter component. Then, with $w_{\cG}$ selected according to \eqref{eq: min-max}, the optimal rule that minimizes the weighted sum of semiparametric efficiency bounds in \eqref{eq: weighted efficiency bound} is also the solution of the maximin problem \eqref{eq: maximin}.
	Theorem \ref{thm: finite-dimensional equiv} contributes to the literature by providing a tractable way to solve the complex infinite-dimensional maximin problem \eqref{eq: maximin}. By applying Theorem \ref{thm: finite-dimensional equiv}, the problem \eqref{eq: maximin} reduces to the more manageable finite-dimensional minimization problem \eqref{eq: min-max}. This simplification significantly eases the optimization process and facilities the data-based estimation. 
	
	Both $\rho_{\cC}(\cdot)$ and $\rho_{\cG}(\cdot)$ reduce to the optimal sampling rule $\rho_{\cS}(\cdot)$ in Theorem \ref{thm: optimal probability} when $d = 1$, but they differ when $d > 1$. As $\rho_{\cG}(\cdot)$ solves the maximin problem \eqref{eq: maximin}, it dominates $\rho_{0}(\cdot)$ whenever there exists a sampling rule that dominates $\rho_{0}(\cdot)$.
	The sampling rule $\rho_{\cG}(\cdot)$ maximizes the minimal relative  improvement compared to the benchmark sampling rule, which makes it promising for estimating a multi-dimensional parameter.
	
	In contrast to sampling rules based on other optimality criteria such as A-optimality or D-optimality \citep{kiefer1959optimum}, the maximin sampling rules $\rho_{\cC}(\cdot)$ and $\rho_{\cG}(\cdot)$ ensure improved efficiency for each parameter component relative to the benchmark rule $\rho_{0}(\cdot)$. This property is particularly desirable when it is hard to determine the importance of the estimation efficiency for each parameter component at the design stage, making $\rho_{\cC}(\cdot)$ and $\rho_{\cG}(\cdot)$ appealing choices in epidemiological and genetic studies with multiple biomarkers, traits, or a multi-level treatment.
	
	\begin{remark}\label{remark: priority}
		The formulation $M(\rho)$ and \eqref{eq: maximin} the case where all the parameters are equally important. In practice, different parameters can have different levels of importance or interest. The proposed
		framework can be extended to take the different importance into consideration.  A method is to add weights to the objective function (3.6). That is, we consider the optimization problem
		\begin{equation}\label{eq: maximin modify}
			\max_{\rho \in \cP_{\cG}} \min_{j=1,\dots,d} \left\{a_{j}^{-1}b_{j}^{-1}\left(\xi_{j} - E\left[\frac{\sigma_{j}^{2}(\bV)}{\rho(\bV)}\right]\right)\right\}
		\end{equation}
		instead of (3.6), where $a_{j} > 0$ is a specified weight to reflect the importance of the $j$th parameter component for $j = 1, \dots, d$. Note that the optimization problem is unchanged if we multiply each of $\{a_{j}\}_{j}^{d}$ by a positive constant. Thus, without loss of generality, we can assume that $\sum_{j=1}^{d}a_{j} = 1$. Let 
		\[
		\cP_{\rm NL} = \left\{
		\begin{aligned}
			&\rho(\cdot):  0\leq \rho(\cdot)\leq 1, \ E[\rho(\bV)] \leq \varpi, \ \text{the semiparametric efficiency bound for}\\
			&\text{every parameter component under $\rho(\cdot)$ is no larger than that under $\rho_{0}(\cdot)$}
		\end{aligned}\right\}.
		\] 
		Then, it is not hard to verify that the solution of \eqref{eq: maximin modify} belongs to $\cP_{\rm NL}$, which is a desirable property. Moreover, this implies that the solutions of \eqref{eq: maximin modify} and		\begin{equation}\label{eq: maximin modify constraint}
			\max_{\rho \in \cP_{\rm NL}} \min_{j=1,\dots,d} \left\{a_{j}^{-1}b_{j}^{-1}\left(\xi_{j} - E\left[\frac{\sigma_{j}^{2}(\bV)}{\rho(\bV)}\right]\right)\right\}
		\end{equation}
		are identical.
		For any $\rho(\cdot) \in \cP_{\rm NL}$, a larger $a_{j}$ makes it more likely that 
		\[
		\min_{j^{\prime}=1,\dots,d} \left\{a_{j^{\prime}}^{-1}b_{j^{\prime}}^{-1}\left(\xi_{j^{\prime}} - E\left[\frac{\sigma_{j^{\prime}}^{2}(\bV)}{\rho(\bV)}\right]\right)\right\} = a_{j}^{-1}b_{j}^{-1}\left(\xi_{j} - E\left[\frac{\sigma_{j}^{2}(\bV)}{\rho(\bV)}\right]\right)
		\]
		and hence prioritizes the efficiency for estimating the $j$th parameter component in the max-min optimization \eqref{eq: maximin modify constraint}. Thus, if the $j$th parameter component is important in a research, one can set $a_{j}$ to be large to prioritize the efficiency of the estimator for the $j$th parameter component. According to such intuition, users have the flexibility to set the parameters $\{a_{j}\}_{j = 1}^{d}$ to prioritize some parameter components based on their research interest. 
		Similar arguments as in the proof of Theorem 2 can show that  the optimal sampling rule that minimizes \eqref{eq: maximin modify} is 
		\[
		\begin{aligned}
			\rho(\cdot; \sigma_{w_{a}}, \tau_{w_{a}})
			= 1\{\sigma_{w_{a}}(\cdot) >  \tau_{w_{a}}\} + \frac{1}{\tau_{w_{a}}}\sigma_{w_{a}}(\cdot)1\{\sigma_{w_{a}}(\cdot) \leq  \tau_{w_{a}}\},
		\end{aligned}
		\]
		where $w_{a} = (w_{a, 1}, \dots, w_{a, d})^{\T}$ is the solution of the problem
		\begin{equation*}
			\min_{w\in \cW_{a}^{\dag}} \left\{\sum_{j=1}^{d}w_{j}b_{j}^{-1}\xi_{j} - E\left[\sigma_{w}(\bV)\max\{\sigma_{w}(\bV), \tau_{w} \}\right]\right\},
		\end{equation*}
		$\sigma_{w}(\cdot) = \sqrt{\sum_{j=1}^{d}w_{j}b_{j}^{-1}\sigma_{j}^{2}(\cdot)}$, $\tau_{w}$ is the unique solution of $E[\rho(\bV;\sigma_{w},\tau)] = \varpi$ with respect to $\tau$, and $\cW_{a}^{\dag}=\{w=(w_{1}, \dots, w_{d}) :  \sum_{j=1}^{d}a_{j}w_{j} = 1, \ 0\leq w_{j} \leq a_{j}^{-1}, \ \text{for} \ j=1,\dots,d\}$. See Section \ref{subsec: sim prioritize} in Supplementary Material for numerical demonstrations on the effect of $\{a_{j}\}_{j = 1}^{d}$ in prioritizing the efficiency of some parameter component.
	\end{remark}
	
	%
	
	\section{Estimation of the optimal sampling rule}\label{sec: est}
	
	The optimal sampling rules proposed in Section \ref{sec: scalar} depend on the conditional variance of the full data efficient influence function, which is unknown.   
	In this section, we consider the estimation of the optimal sampling rule, while the parameter of interest can be either scalar or multi-dimensional.  
	We draw a simple random subsample at the beginning of the second phase, and use the pilot sample to estimate the optimal sampling rule. And then we draw the remaining subjects according to the estimated rule. 	
	Specifically, let $(\bV_{1},\bU_{1}),\dots,(\bV_{n},\bU_{n})$ be independent and identically distributed (i.i.d.) copies of $(\bV, \bU)$. For $i=1,\dots,n$, $\bV_{i}$ is observed in the first phase. Let $R_{1i}$, $i=1, \dots, n$, be independent Bernoulli random variables which equal to one with probability $\kappa_{n}$, where $\kappa_{n}$ is the predetermined expected proportion of the pilot sample that may change with $n$. The $i$th subject is included in the pilot sample, and $\bU_{i}$ is measured if $R_{1i} = 1$, and not otherwise. Then the optimal sampling rules can be estimated based on the pilot sample. Subsequently, for each subject with $R_{1i}=0$, we generate an independent Bernoulli random variables $R_{2i}$ which equals one with probability $\tilde{\rho}(\bV_{i})$, where $\tilde{\rho}(\cdot)$ is the estimated sampling rule. Afterwards, $\bU_{i}$ is measured for subjects with $R_{2i} = 1$. Based on the data collected in the two phases and the pilot sample, $\theta_{0}$ can be estimated using standard semiparametric methods, such as the targeted maximum likelihood \citep{van2006targeted}, the estimating equation method \citep{tsiatis2007semiparametric}, and the one-step estimation method \citep{bickel1982adaptive}.  For illustration, we state the one-step estimation procedure under two-phase designs in Section \ref{subsec: one-step} in Supplementary Material. Next, we get down to the estimation of the sampling rule. The sampling rule for $U$ in the above procedure is a mixture of the uniform rule and the estimated optimal rule due to the presence of the pilot sample. This sampling rule is nearly optimal provided the $\tilde{\rho}(\cdot)$ is consistent for the optimal sampling rule and the proportion of pilot sample is small in the sense that $\kappa_{n} / \varpi \to 0$.
	
	We first consider the estimation of the optimal sampling rule $\rho_{j}(\cdot)$ for each component. The optimal sampling rule $\rho_{\cS}(\cdot)$ for a scalar parameter given in Theorem \ref{thm: optimal probability} can be regarded as a special case of $\rho_{j}(\cdot)$ with $j = d= 1$. Recall that $\rho_{j}(\cdot) = \rho(\cdot; \sigma_{j}, \tau_{j})$ where $\tau_{j}$ satisfies $E[\rho(V; \sigma_{j}, \tau_{j})] = \varpi$. The unknown quantity $\sigma_{j}(\cdot)$ can be estimated by well-established statistical methods based on the observations in the first phase and the pilot sample for $j = 1,\dots,d$. For more details, see Section \ref{subsec: MVR} in Supplementary Material.
	Denote the estimate for $\Pi_{j}(\cdot)$ and $\sigma_{j}(\cdot)$ by $\widetilde{\Pi}_{j}(\cdot)$ and $\tilde{\sigma}_{j}(\cdot)$, respectively, for $j = 1,\dots, d$. The threshold $\tau_{j}$ is another quantity requiring estimating. Note that we aim to take $\varpi$ as the expected proportion of subjects included in the second phase, and the pilot sample occupies an expected proportion of $\kappa_{n}$. Thus, the expected proportion of subjects included in the second phase with $R_{1i} = 0$ should be $\varpi - \kappa_{n}$. Hence we estimate the threshold $\tau_{j}$ by the solution $\tilde{\tau}_{j}$ of the equation
	\begin{equation}\label{eq: threshold estimating equation}
		\frac{1}{n}\sum_{i=1}^{n}(1 -R_{1i})\left(1\{\tilde{\sigma}_{j}(\bV_{i}) > \tau\} + \frac{\tilde{\sigma}_{j}(\bV_{i})}{\tau}1\{\tilde{\sigma}_{j}(\bV_{i}) \leq \tau\}\right) = \varpi - \kappa_{n}.
	\end{equation}
	Then $\rho_{j}(\cdot)$ can be estimated by
	\[\tilde{\rho}_{j}(\cdot) = 1\{\tilde{\sigma}_{j}(\cdot) > \tilde{\tau}_{j}\} + \frac{\tilde{\sigma}_{j}(\cdot)}{\tilde{\tau}_{j}}1\{\tilde{\sigma}_{j}(\cdot) \leq \tilde{\tau}_{j}\},\]
	for $j = 1,\dots, d$. 
	The sampling rule $\rho_{\cC}(\cdot)$ for a multi-dimensional parameter can be estimated by $\tilde{\rho}_{\cC}(\cdot) = \rho_{0}(\cdot) + \sum_{j=1}^{d}\widetilde{w}_{{\cC},j}(\tilde{\rho}_{j}(\cdot) - \rho_{0}(\cdot))$, where 
	\[
	\begin{aligned}
		&\widetilde{w}_{\cC} = (\widetilde{w}_{{\cC}, 0},\dots, \widetilde{w}_{{\cC}, d})\\
		& = \mathop{\arg\max}_{w \in \cW} 
		\min_{j=1,\dots,d} \left\{\tilde{b}_{j}^{-1}\left(\tilde{\xi}_{j}-\frac{1}{n}\sum_{i=1}^{n} \frac{\tilde{\sigma}_{j}^{2}(\bV_{i})}{\rho_{0}(\bV_{i}) + \sum_{j=1}^{d}w_{j}(\tilde{\rho}_{j}(\bV_{i}) - \rho_{0}(\bV_{i}))}\right)\right\}, 
	\end{aligned}
	\]
	$\tilde{\xi}_{j} = n^{-1}\sum_{i=1}^{n}\tilde{\sigma}_{j}^{2}(\bV_{i})/\rho_{0}(\bV_{i})$ and 
	$\tilde{b}_{j} = n^{-1}\sum_{i=1}^{n}\tilde{\sigma}_{j}^{2}(\bV_{i})/\rho_{0}(\bV_{i}) + n^{-1}\sum_{i=1}^{n}\{\widetilde{\Pi}_{j}(V_{i}) - n^{-1}\sum_{i=1}^{n}\widetilde{\Pi}_{j}(V_{i})\}^{2}$ are the estimates of $\xi_{j}$ and $b_{j}$, for $j = 1,\dots, d$.
	According to Theorem \ref{thm: finite-dimensional equiv}, to estimate $\rho_{\cG}(\cdot)$, we first estimate $w_{\cG}$ by
	\[\widetilde{w}_{\cG} = \mathop{\arg\min_{w\in \cW^{\dag}}} \left\{\sum_{j=1}^{d}w_{j}\tilde{b}_{j}^{-1}\tilde{\xi}_{j} - \frac{1}{n}\sum_{i=1}^{n}\tilde{\sigma}_{w}(\bV_{i})\max\{\tilde{\sigma}_{w}(\bV_{i}), \tilde{\tau}_{w} \}\right\},\]
	where $\tilde{\sigma}_{w}(\cdot) = \sqrt{\sum_{j=1}^{d}w_{j}\tilde{b}_{j}^{-1}\tilde{\sigma}_{j}^{2}(\cdot)}$ and $\tilde{\tau}_{w}$ is the solution of \eqref{eq: threshold estimating equation} with $\tilde{\sigma}_{j}(\cdot)$ replaced by $\tilde{\sigma}_{w}(\cdot)$. Then the estimator for $\rho_{\cG}(\cdot)$ is given by
	\[\tilde{\rho}_{\cG}(\cdot) = 1\{\tilde{\sigma}_{\widetilde{w}_{\cG}}(\cdot) > \tilde{\tau}_{\widetilde{w}_{\cG}}\} + \frac{\tilde{\sigma}_{\widetilde{w}_{\cG}}(\cdot)}{\tilde{\tau}_{\widetilde{w}_{\cG}}}1\{\tilde{\sigma}_{\widetilde{w}_{\cG}}(\cdot) \leq \tilde{\tau}_{\widetilde{w}_{\cG}}\},\]
	where $\tilde{\sigma}_{\widetilde{w}_{\cG}}(\cdot)$ is $\tilde{\sigma}_{w}(\cdot)$ with $w$ equaling to $\widetilde{w}_{\cG}$.
	The following theorem shows that the above sampling procedure satisfies the exact budget constraint, i.e., the expected proportion of subjects included in the second phase is exactly $\varpi$.
	\begin{theorem}\label{thm: exact budget constraint}
		If $\tilde{\rho}(\cdot) = \tilde{\rho}_{\cC}(\cdot)$, $\tilde{\rho}_{\cG}(\cdot)$, or $\tilde{\rho}_{j}(\cdot)$ ($j\in\{1, \dots, d\}$), we have 
		\[E\left[\frac{1}{n}\sum_{i=1}^{n}(R_{1i} + R_{2i})\right] = \varpi.\]
	\end{theorem}	
	The proof of Theorem \ref{thm: exact budget constraint} is provided in Section \ref{subsec: proof exact budget} in Supplementary Material.
	We next investigate the convergence rate of the proposed sampling rule estimators, starting with the component-wise optimal sampling rule $\tilde{\rho}_{j}(\cdot)$ for $j = 1, \dots, d$. For any function $f(\cdot)$, let $\|f\|_{\infty} = \sup_{v}|f(v)|$ be the infinity norm of $f$. To establish the theoretical result, it is required that the size of the pilot sample and the conditional standard deviation estimator satisfy the following condition.
	\begin{condition}\label{cond: rate sigma}
		$\kappa_{n}\to 0$, $n\kappa_{n} \to \infty$, $\|\widetilde{\Pi}_{j} - \Pi_{j}\|_{\infty} = O_{P}\left\{(n\kappa_{n})^{-\delta}\right\}$ and 
		$\|\tilde{\sigma}_{j} - \sigma_{j}\|_{\infty} = O_{P}\left\{(n\kappa_{n})^{-\delta}\right\}$ for $j = 1,\dots,d$ and some $0<\delta \leq 1/2$.
	\end{condition}
	The first two rate conditions can be satisfied by properly chosen $\kappa_{n}$. In practice, we recommend to take $\kappa_{n}$ as a sequence that converges slowly to zero, e.g., $\kappa_{n} = \varpi / \log(\varpi n)$, to ensure the stability of the estimates based on the pilot sample. 
	Convergence rates required by Condition \ref{cond: rate sigma} are well established in the nonparametric estimation literature. Specifically, the convergence rate is readily available under certain regularity conditions if the conditional variance is estimated using the kernel method \citep{hardle1988strong} or the sieve method \citep{chen2015optimal}. Under suitable regularity conditions, we have the following theorem. See Section \ref{subsec: proof of convergence rule comp} of Supplementary Material for the proof of the theorem.
	\begin{theorem}\label{thm: convergence rule comp}
		Suppose the parameter $\theta_{0}$ is scalar. Under Condition \ref{cond: rate sigma} and Conditions \ref{cond: id comp}--\ref{cond: bound cond mean-var} in Supplementary Material, we have
		\[\|\tilde{\rho}_{\cS} - \rho_{\cS}\|_{\infty} = O_{P}\left\{(n\kappa_{n})^{-\delta} + \kappa_{n}\right\}.
		\]
	\end{theorem}
	Theorem \ref{thm: convergence rule comp} can be directly extended to prove $\|\tilde{\rho}_{j} - \rho_{j}\|_{\infty} = O_{P}\left\{(n\kappa_{n})^{-\delta} + \kappa_{n}\right\}$ for $j = 1,\dots,d$ when the parameter is multi-dimensional.
	From the proof of Theorem \ref{thm: convergence rule comp}, it can be seen that the term $\kappa_{n}$ in the convergence rate can be removed if $(1 - \kappa_{n})\varpi$ is used instead of $\varpi - \kappa_{n}$ on the right-hand side of \eqref{eq: threshold estimating equation}. However, we use $\varpi - \kappa_{n}$ in \eqref{eq: threshold estimating equation} to ensure the exact budget constraint stated in Theorem \ref{thm: exact budget constraint}.
	For the estimation of $\rho_{\cC}(\cdot)$ and $\rho_{\cG}(\cdot)$, we have the following convergence result.
	\begin{theorem}\label{thm: convergence rule mv}
		Suppose $\rho_{0}(\cdot)$ is bounded away from zero. Then under Condition \ref{cond: rate sigma} and Conditions \ref{cond: id comp}--\ref{cond: local convexity Gm} in Supplementary Material, we have 
		\[\|\tilde{\rho}_{\cC} - \rho_{\cC}\|_{\infty} = O_{P}\left\{\sqrt{(n\kappa_{n})^{-\delta} + \kappa_{n}}\right\} \ \text{and} \ \|\tilde{\rho}_{\cG} - \rho_{\cG}\|_{\infty} = O_{P}\left\{\sqrt{(n\kappa_{n})^{-\delta} + \kappa_{n}}\right\}.\]
	\end{theorem}
	The proof of Theorem \ref{thm: convergence rule mv} can be found in Section \ref{subsec: proof of thm rule mv} of Supplementary Material.

		\section{Simulation study}\label{sec: sim}
		In this section, we use ``uni" to denote the uniform rule, ``S-opt" to denote the optimal sampling rule for a scalar parameter $\tilde{\rho}_{\cS}$, ``C-opt" to denote the sampling rule $\tilde{\rho}_{\cC}$, and ``G-opt" to denote the sampling rule $\tilde{\rho}_{\cG}$. 
		We consider the average treatment effect estimation problem introduced in Example \ref{eg: ATE}. The problem is akin to the causal effect estimation problem in observational studies with missing confounders discussed in \cite{lin2014adjustment} and \cite{yang2019combining}. Let $\bZ$ be a $q$-dimensional covariate vector  with independent $U[-2.5,2.5]$ components, where $U[-2.5, 2.5]$ is the uniform distribution on $[-2.5, 2.5]$. Suppose the confounder $X$ follows the model
		\[X = \bzeta_{q}^{\T}\bZ + \epsilon_{x},\]
		where $\bzeta_{q} = (0.5/\sqrt{q}, \dots, 0.5/\sqrt{q})^{\T}$ and $\epsilon_{x}$ follows a normal distribution with mean zero and variance $0.25$.
		The potential outcomes $Y_{0}$ and $Y_{1}$ follow the model
		\[
		\begin{aligned}
			Y_{0} &= 0.5\bzeta_{q}^{\T}\bZ + X + \exp(2\bzeta_{q}^{\T}\bZ)\epsilon_{y},\\
			Y_{1} &= \theta_{0} + 0.5\bzeta_{q}^{\T}\bZ - X + \exp(2\bzeta_{q}^{\T}\bZ)\epsilon_{y}.
		\end{aligned}
		\]
		where $\epsilon_{y}$ is a standard normal error. Under the above model, $\theta_{0} = E[Y_{1} - Y_{0}]$ is the average treatment effect. We take $\theta_{0} = 1.5$ in the simulation. The treatment indicator $T$ is a Bernoulli variable that is equal to one with probability $1/\{1 + \exp(0.1\bzeta_{q}^{\T}\bZ - 0.5X)\}$. Then the outcome variable is $Y = TY_{1} + (1-T)Y_{0}$. 
		Let $\bV = (Y, T, \bZ^{\T})^{\T}$ be the vector of first-phase variables and $\bU = X$ the second-phase variable in this example. In the simulation, we generate $n = 2000$ or $5000$  i.i.d. copies of $(V, U)$, draw a pilot sample and estimate the sampling rules based on the nonparametric method suggested in Section \ref{subsec: MVR} in Supplementary Material. We apply different sampling rules to draw the sample in the second phase. Then, the one-step estimation method is adopted to estimate $\theta_{0}$. The one-step estimation method \citep{bickel1982adaptive} is a classic semiparametric method to construct an efficient estimator based on the efficient influence function. Please refer to Section \ref{subsec: one-step} in Supplementary Material for more details. We consider $q = 1$ and $5$ and take the proportion of pilot sample $\kappa_{n} = \varpi/[1 + \log\{\varpi n/(q + 1)\}]$. The quantity $\varpi/[1 + \log\{\varpi n/(q + 1)\}]$ increases as $q$ grows and decreases to zero slowly as $n\to \infty$. Thus, the choice of $\kappa_{n}$ ensures Condition \ref{cond: rate sigma} while taking the dimension of the first-phase variable vector into consideration. 
		Figure~\ref{fig: ATE boxplot} presents boxplots of the average treatment effect estimator over $500$ simulations with $n = 2000, 5000$ and $q = 1,5$.
		\begin{figure}[h]
			\centering
			\subcaptionbox{}{\includegraphics[scale=0.35]{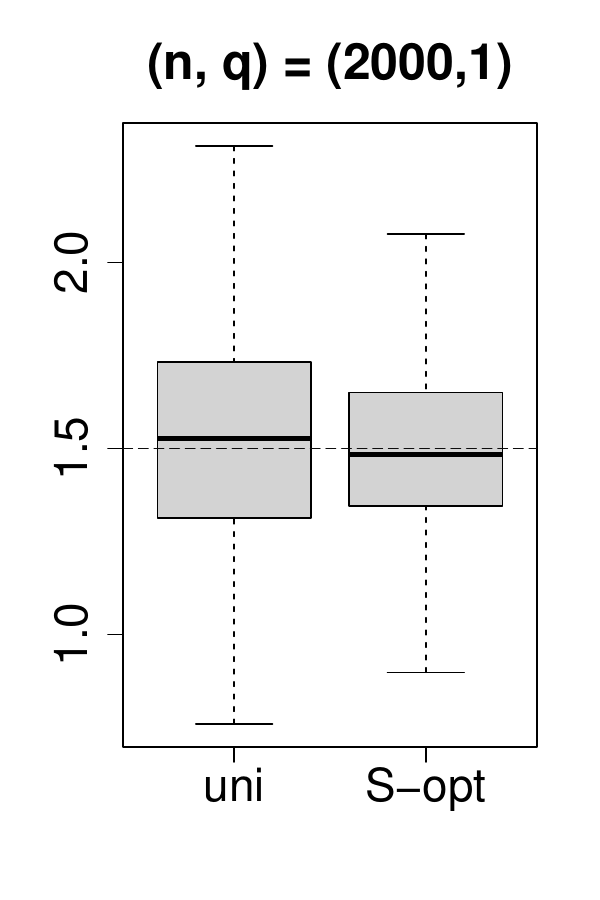}}
			\subcaptionbox{}{\includegraphics[scale=0.35]{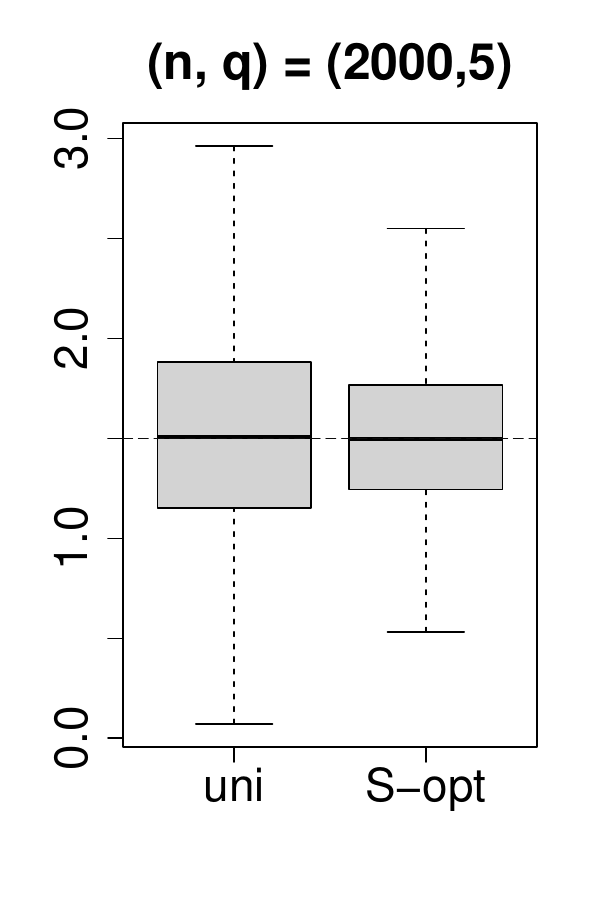}}
			\subcaptionbox{}{\includegraphics[scale=0.35]{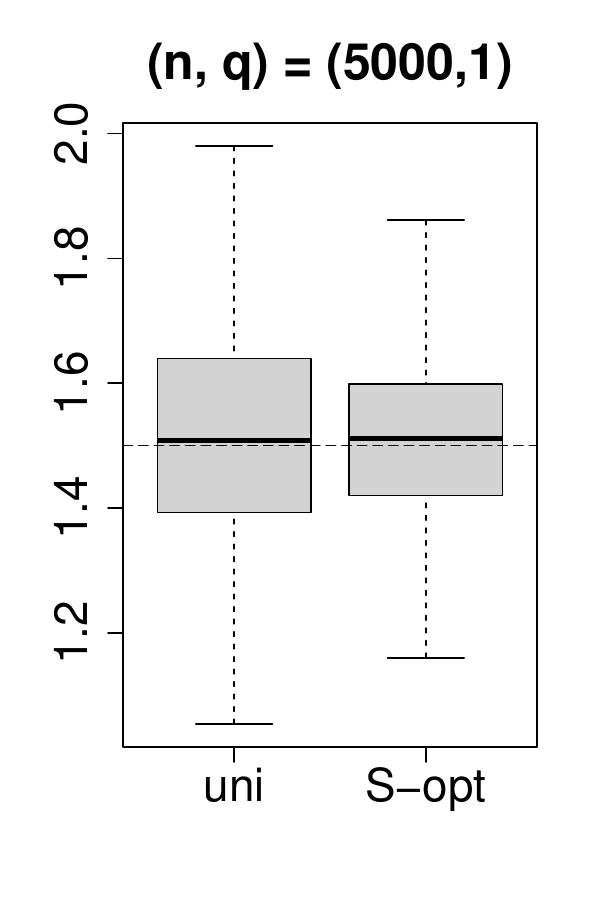}}
			\subcaptionbox{}{\includegraphics[scale=0.35]{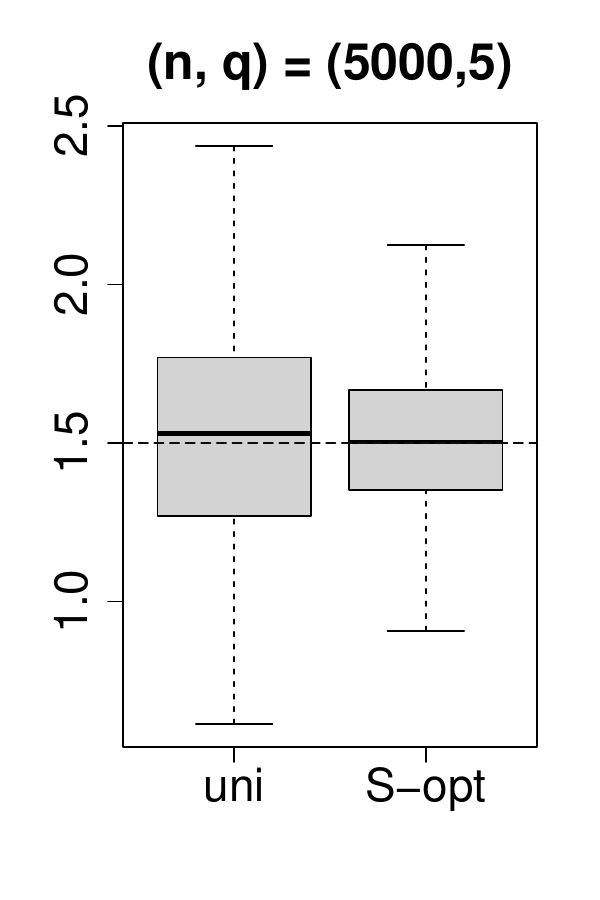}}
			\caption{Boxplots for the average treatment effect estimation with different combinations of $n$ and $q$; dashed lines are the true values.}\label{fig: ATE boxplot}
		\end{figure}
		Figure~\ref{fig: ATE boxplot} shows that the efficiency of the average treatment effect estimator is improved under $\tilde{\rho}_{\cS}$ compared to the uniform rule. The improvement is consistently observed across different combinations of $n$ and $q$. We compute the relative efficiency (RE) under $\tilde{\rho}_{\cS}$ with respect to the uniform rule. The RE of a sampling rule is the variance ratio of the estimator under the uniform rule to that under the considered sampling rule. Higher RE indicates greater improvement than the uniform rule. The REs are $1.6995$, $2.1835$, $1.8907$ and $2.3239$ for $(n, q) = (2000, 1), (2000, 5), (5000, 1)$, and $(5000, 5)$, respectively, which indicates significant efficiency improvements across different scenarios.
		
		In practice, there may be multiple treatments. In this case, the average treatment effects of different treatments compared to the control are of interest. Next, we consider the average treatment effect estimation with two different treatments (denoted by ``$1$'', and ``$2$'') and a control group (denoted by ``$0$"). Define $\bZ$ in the same way as in the scalar case. Suppose $X$ follows the model
		\[X = \bzeta_{q}^{\T}\bZ + \epsilon_{x},\]
		where $\epsilon_{x}$ follows a normal distribution with mean zero and variance $0.25$.
		The potential outcomes are generated from
		\begin{align*}
			&Y_{0} = 0.5\bzeta_{q}^{\T}\bZ + X + \{0.5 + \nu_{2}(\bZ)\}\epsilon_{y}, \\
			&Y_{1} = \theta_{01} + 0.5\bzeta_{q}^{\T}\bZ - X + \{\nu_{1}(\bZ) + \nu_{2}(\bZ)\}\epsilon_{y},\\
			&Y_{2} = \theta_{02} - 0.5\bzeta_{q}^{\T}\bZ - 0.5X + \nu_{2}(\bZ)\epsilon_{y},
		\end{align*} 
		where $\theta_{01} = 1$, $\theta_ {02} = 0.5$, $\epsilon_{y}$ follows a standard normal distribution, 
		$\nu_{1}(z) = \sqrt{0.1 + (2\bzeta_{q}^{\T}z)^{4}}$, and $\nu_{2}(z) = \exp(2\bzeta_{q}^{\T}z)$.
		It is easy to verify that the average treatment effects $E[Y_{1} - Y_{0}]$ and $E[Y_{2} - Y_{0}]$ equal to $\theta_{01}$ and $\theta_{02}$, respectively.
		Suppose the treatment indicator $T$ equals to $0$, $1$, and $2$ with probabilities $1 - p_{1}(X, \bZ) - p_{2}(X, \bZ)$, $p_{1}(X, \bZ)$ and $p_{2}(X, \bZ)$, where
		\[
		\begin{aligned}
			p_{1}(X, \bZ) & = \frac{\exp(-0.1\bzeta_{q}^{\T}\bZ + 0.25X)}{1+\exp(-0.1\bzeta_{q}^{\T}\bZ + 0.25X) + \exp(0.1\bzeta_{q}^{\T}\bZ - 0.25X)},\\
			p_{2}(X, \bZ) &= \frac{\exp(0.1\bzeta_{q}^{\T}\bZ - 0.25X)}{1+\exp(-0.1\bzeta_{q}^{\T}\bZ + 0.25X) + \exp(0.1\bzeta_{q}^{\T}\bZ - 0.25X)}.
		\end{aligned}
		\]
		The observed outcome $Y$ equals to $Y_{t}$ if $T = t$ for $t = 0, 1,$ and $2$.
		Table~\ref{table: mv ATE} presents the bias, standard error (SE) of the estimator, and the RE compared to the uniform rule based on $500$ simulations.
		\begin{table}
			\centering
			\def~{\hphantom{0}}
			\caption{Bias, SE, and RE in two-dimensional average treatment effect estimation}
			\begin{tabular}{llcccccc}
				\multirow{2}{*}{$(n, q)$}&\multirow{2}{*}{Rule} & \multicolumn{3}{c}{Estimate of $\theta_{01}$}  & \multicolumn{3}{c}{Estimate of $\theta_{02}$}\\
				& & Bias & SE & RE & Bias & SE& RE \\
				\specialrule{0em}{-4pt}{-4pt}\\
				\multirow{3}{*}{$(2000, 1)$}
				&uni & -0.0055 & 0.5630 & 1.0000 & 0.0150 & 0.4755 & 1.0000\\
				&C-opt & 0.0022 & 0.4558 & 1.5257 & 0.0220 & 0.3620 & 1.7254\\
				&G-opt & 0.0157 & 0.4597 & 1.4999 & 0.0237 & 0.3596 & 1.7485\\
				\specialrule{0em}{-4pt}{-4pt}\\
				\multirow{3}{*}{$(2000, 5)$}
				&uni &  0.0220 & 0.9095 & 1.0000 & -0.0417 & 0.8307 & 1.0000\\
				&C-opt & -0.0171 & 0.7734 & 1.3829 & -0.0545 & 0.6452 & 1.6577\\
				&G-opt & -0.0206 & 0.7384 & 1.5171 & -0.0354 & 0.6199 & 1.7957\\
				\specialrule{0em}{-4pt}{-4pt}\\
				\multirow{3}{*}{$(5000, 1)$}
				&uni & 0.0200 & 0.3349 & 1.0000 & -0.0043 & 0.2765 & 1.0000\\
				&C-opt & 0.0289 & 0.2564 & 1.7061 & -0.0010 & 0.2179 & 1.6102\\
				&G-opt & 0.0189 & 0.2459 & 1.8549 & -0.0043 & 0.2113 & 1.7123\\
				\specialrule{0em}{-4pt}{-4pt}\\
				\multirow{3}{*}{$(5000, 5)$}
				&uni & -0.0353 & 0.5363 & 1.0000 & -0.0119 & 0.4667 & 1.0000\\
				&C-opt & -0.0161 & 0.3970 & 1.8249 & -0.0181 & 0.3332 & 1.9618\\
				&G-opt & -0.0006 & 0.3836 & 1.9546 & -0.0110 & 0.3115 & 2.2447\\
			\end{tabular}\label{table: mv ATE}
		\end{table}
		Table~\ref{table: mv ATE} shows that the SEs are smaller under the sampling rules $\tilde{\rho}_{\cC}$ and $\tilde{\rho}_{\cG}$ than under the uniform rule. Moreover, the improvement is consistent and significant across different parameter components and values of $(n,q)$. In some cases, the RE is close to or even larger than two. The improvement of $\tilde{\rho}_{\cG}$ tends to be greater than $\tilde{\rho}_{\cC}$ in most cases, probably because $\tilde{\rho}_{\cG}$ estimates the optimal solution over the full space, while $\tilde{\rho}_{\cC}$ estimates the optimal solution over only a subspace. 
		In addition, the simulation results in this section indicate that the proposed method's performance remains stable even as the dimension of the first-phase variable increases. This suggests that the proposed method is not highly sensitive to the dimensionality of the first-phase variable despite the involvement of nonparametric methods in the estimation procedure.
		

		\section{Real data analysis}\label{sec: real data}
		In this section, we illustrate the proposed approach using AIDS Clinical Trials Group (ACTG) 175 data \citep{hammer1996}, which can be found in the R package \emph{speff2trial}. The ACTG175 clinical trial evaluates four therapies for treating human immunodeficiency virus infection. The four treatments are: $1$, zidovudine; $2$, zidovudine and didanosine; $3$, zidovudine and zalcitabine; and $4$, didanosine. In the trial, $2139$ subjects were assigned randomly to the four treatment groups. Covariates such as age, gender, weight, medication history, and symptom indicator were recorded before the treatment, and CD4 and CD8 cell counts were measured after $20$ weeks of the treatment. The CD4/CD8 ratio describes the overall immune function of an individual. A low value of CD4/CD8 ratio indicates a severe immune dysfunction. We use the CD4/CD8 ratio as the outcome variable and consider the average treatment effects of the three zidovudine-treated groups (groups $1$, $2$, and $3$) compared to the non-zidovudine-treated group (group $4$). Let $\theta_{0j}$ be the average treatment effect of the treatment $j$ compared to the  treatment $4$ for $j = 1,\dots, 3$.
		
		The original experiment measured CD4 and CD8 cell counts for all $2139$ subjects after $20$ weeks of the treatment. However, it can be costly to measure the CD4 and CD8 cell counts for more than two thousand people. The cost of the trials would be drastically reduced if a two-phase design were adopted and the CD4 and CD8 cell counts were measured for only a subset of patients. We applied the proposed method to design the second-phase sampling rule and obtained results as if the experiment was conducted under the two-phase design.
		We use six variables as the first-phase variables $\bV$, including age, gender, weight, medication history, symptom indicator, and treatment assignment. The outcome is considered as the second-phase variable $\bU$. In the implementation, we take $\varpi = 0.3$ and $\kappa_{n} = \varpi/[1 + \log(0.1n\varpi)]$.
		A pilot sample is randomly drawn with inclusion probability $\kappa_{n}$. The sampling rules $\tilde{\rho}_{\cC}$ and $\tilde{\rho}_{\cG}$ are constructed based on the pilot sample and the full data efficient influence function in Example \ref{eg: ATE} in Section \ref{sec: examples of EIF}  in Supplementary Material. Then, a sample is drawn according to the resulting sampling rule. The full data efficient influence function involves the propensity score and the outcome regression function of each group. The propensity score for each group is $1/4$ because subjects were randomly assigned to each group with equal probability. We adopt a linear model for the outcome regression function.  
		
		We employ the one-step estimation method to estimate the average treatment effects. In line with the simulation, ``uni", ``C-opt", and ``G-opt" represent the uniform rule, the sampling rule $\tilde{\rho}_{\cC}$, and the sampling rule $\tilde{\rho}_{\cG}$, respectively. Table~\ref{table: mv ACTG175} provides the estimates for $\theta_{01}$, $\theta_{02}$, and $\theta_{03}$ under different sampling rules. Additionally, it includes the estimated SEs based on the asymptotic variance of the estimator and p-values associated with the Wald tests for testing whether each average treatment effect equals to zero.
		\begin{table}[h]
			\centering
			\caption{Point estimates (Est), estimated SEs, and p-values associated with the estimation of $\theta_{01}$, $\theta_{02}$, and $\theta_{03}$ in the ACTG175 trial under different sampling rules}
			\resizebox{\columnwidth}{!}{
				\begin{tabular}{l*{9}{c}}
					\multirow{2}{*}{Rule} & \multicolumn{3}{c}{Estimation of $\theta_{01}$}  & \multicolumn{3}{c}{Estimation of $\theta_{02}$}& \multicolumn{3}{c}{Estimation of $\theta_{03}$}\\
					& Est & SE & p-value & Est & SE & p-value & Est & SE & p-value\\
					\specialrule{0em}{-4pt}{-4pt}\\
					uni & -0.0489 & 0.0223 & 0.0283 & 0.0166 & 0.0247 & 0.5013 & -0.0034 & 0.0228 & 0.8805\\
					C-opt & -0.0606 & 0.0211 & 0.0041 & 0.0182 & 0.0235 & 0.4398 & -0.0012 & 0.0216 & 0.9574\\
					G-opt & -0.0630 & 0.0211 & 0.0029 & 0.0147 & 0.0231 & 0.5227 & -0.0006 & 0.0216 & 0.9794\\
				\end{tabular}
			}\label{table: mv ACTG175}
		\end{table} 
		Table~\ref{table: mv ACTG175} shows that the estimators for $\theta_{01}$, $\theta_{02}$, and $\theta_{03}$ all have smaller estimated SEs under $\tilde{\rho}_{\cC}$ and $\tilde{\rho}_{\cG}$ compared to the uniform rule. The p-values under $\tilde{\rho}_{\cC}$ and $\tilde{\rho}_{\cG}$ suggest that, at the $0.01$ significance level, treatment $1$ performs differently than treatment $4$ in increasing the CD4/CD8 ratio, while there is no significant difference between treatments $2$, $3$, and treatment $4$. In contrast, the results under the uniform rule discover no significant difference among treatments $1$, $2$, $3$, and treatment $4$ under the significance level $0.01$.
		For comparison,	we calculate the full-cohort augmented inverse probability weighted estimator for $\theta_{01}$, $\theta_{02}$, and $\theta_{03}$ based on the first-phase variables and the outcome of all subjects. It is worth noting that the full-cohort estimator requires measuring the outcome of all $2139$ subjects, while the estimates in Table~\ref{table: mv ACTG175} only needs the outcome of approximately $30\%$ of the subjects. The all data-based estimates are $-0.0500$,  $0.0166$, and $0.0054$ for $\theta_{01}$, $\theta_{02}$ and $\theta_{03}$, respectively, with p-values $0.0005$, $0.2851$, and $0.7140$. These p-values support the same conclusions as those suggested by the estimators under $\tilde{\rho}_{\cC}$ and $\tilde{\rho}_{\cG}$.
		
		\section{Discussion}
		This paper proposes a unified framework for designing sampling rules in two-phase studies. The proposed method offers promising guarantees in estimation efficiency and budget control. It exhibits broad applicability across diverse research domains, such as the electronic health record data analysis \citep{lotspeich2022efficient, zhang2024patient}, genetic association analysis \citep{lin2013quantitative}, causal inference \citep{lin2014adjustment, yang2019combining}, and measurement error problems. Our framework accommodates multi-dimensional parameters, making it well-suited for association studies involving multiple traits and the estimation of causal effects with multiple treatments. The current framework relies on the semiparametric efficiency bound for a fixed-dimensional parameter. An important topic for future research is the development of suitable criteria for high-dimensional problems and the corresponding optimal designs.
		
		The implementation of the proposed design method requires the full-data efficient influence function.
		The efficient influence functions are available from the literature and easy to calculate in many important examples such as the examples presented in Section S1.1 in Supplementary Material, as well as in a lot of other examples \citep{NEWEY1993419, tsiatis2007semiparametric}. 
		However, the derivation and calculation of the full-data efficient influence function may be 
		challenging under some complicated semiparametric models such as those with missing covariates or complex censoring. Additional theoretical or computational efforts are required to implement the proposed method in these scenarios.
		
		\section*{Acknowledgement}
		Wang's research was supported by the National Natural Science Foundation of China (General program 12271510, General
		program 11871460, and program for Innovative Research Group 61621003), and a grant from the Key Lab of Random Complex Structure and Data Science, CAS. Miao's research was supported by the National Key R\&D Program (2022YFA1008100) and the National Natural Science Foundation of China (Genral program 12071015).
		
		\newpage
		\setcounter{section}{0}
		\setcounter{condition}{0}
		\setcounter{theorem}{0}
		\setcounter{table}{0}
		\setcounter{figure}{0}
		\setcounter{example}{0}
		\setcounter{proposition}{0}
		\setcounter{lemma}{0}
		\renewcommand{\thesection}{S\arabic{section}}
		\renewcommand{\thecondition}{S\arabic{condition}}
		\renewcommand{\thetheorem}{S\arabic{theorem}}
		\renewcommand{\thetable}{S\arabic{table}}
		\renewcommand{\thefigure}{S\arabic{figure}}
		\renewcommand{\theexample}{S\arabic{example}}
		\renewcommand{\theproposition}{S\arabic{proposition}}
		\renewcommand{\thelemma}{S\arabic{lemma}}
		{\noindent\LARGE \bf Supplementary Materials}
		\section{Examples}\label{sec: examples}
		\subsection{Examples of full data efficient influence functions}\label{sec: examples of EIF}
		We introduce some examples of full data efficient influence functions for illustration. 
		\begin{example}\label{eg: outcome mean}
			Let $Y$ be a vector of outcomes which is hard to obtain. Suppose the parameter of interest is the outcome mean $\theta_{0} = E[Y]$. Let $\bZ$ be a vector of inexpensive covariates that is predictive to $Y$ and hence useful in estimating $\theta_{0}$. In two-phase studies, one can collect $\bV = \bZ$ in the first phase and measure $\bU = Y$ for a subset of subjects in the second phase. In this example, the full data efficient influence function is $\psi = Y - \theta_{0}$. 
		\end{example}
		
		\begin{example}\label{eg: least squares}
			Let $Y$ be a scalar outcome which is easy to obtain, $\bZ$ a vector of inexpensive covariates, and $\bX$ a vector of expensive covariates. Suppose the parameter of interest is the least squares regression coefficient $\theta_{0}$ of $\bX$ in the regression of $Y$ on $\bZ,\bX$, which is determined by the estimating equation $E[(\bX^{\T},\bZ^{\T})^{\T}(Y - \bX^{\T}\theta_{0} - \bZ^{\T}\beta_{0})] = 0$ where $\beta_{0}$ is the nuisance parameter.
			In two-phase studies, $\bV = (Y,\bZ)$ is collected in the first phase and $\bU = \bX$ is measured for a subset of subjects in the second phase.  In this case, the full data efficient influence function is $\psi = (E[(\bX - \alpha_{0}\bZ)(\bX - \alpha_{0}\bZ)^{\T}])^{-1}(\bX - \alpha_{0}\bZ)(Y - \bX^{\T}\theta_{0} - \bZ^{\T}\beta_{0})$ where $\alpha_{0} = E[\bX \bZ^{\T}](E[\bZ\bZ^{\T}])^{-1}$ is the population linear regression coefficient of $\bX$ on $\bZ$. 
		\end{example}
		
		\begin{example}\label{eg: ATE}
			Let $T\in \{0,1\}$ be a binary treatment indicator, and $Y$ the outcome. Suppose the parameter of interest is the average treatment effect, i.e., $\theta_{0} = E[Y_{1} - Y_{0}]$, where $Y_{1}$ and $Y_{0}$ are the potential outcomes under treatments $``1"$ and $``0"$, respectively. In observational studies, one needs to properly adjust for confounders to estimate $\theta_{0}$ consistently. In practice, some confounders $\bX$ may be hard to measure, while $Y$, $T$, and other confounders $\bZ$ can be easily accessible. Then, a two-phase study can be conducted, where $\bV = (Y, T,\bZ)$ is collected in the first phase, and $\bU = \bX$ is measured for a subset of subjects in the second phase. Under the unconfoundness condition $(Y_{1}, Y_{0}) \Perp T \mid (\bX, \bZ)$, the full data efficient influence function is
			\begin{equation*}
				\begin{aligned}
					\psi& = \frac{TY}{\pi(\bX,\bZ)} - \frac{(1-T)Y}{1-\pi(\bX,\bZ)} - \left\{\frac{T}{\pi(\bX,\bZ)} - 1 \right\}m_{1}(\bX,\bZ) \\
					&\quad +  \left\{\frac{1-T}{1-\pi(\bX,\bZ)} - 1 \right\}m_{0}(\bX,\bZ) - \theta_{0},
			\end{aligned}\end{equation*}
			where $\pi(x,z) = P(T = 1\mid \bX=x, \bZ=z)$ is the propensity score, and $m_{t}(x,z) = E[Y\mid \bX=x, \bZ = z,T=t]$ is the outcome regression function for $t = 0,1$. 
		\end{example}
		The outcome mean estimation in Example \ref{eg: outcome mean} is an important problem in survey sampling \citep{cochran2007sampling} and epidemiological studies \citep{mcnamee2002optimal,gilbert2014optimal}.
		Regression problems with expensive covariates in Example \ref{eg: least squares} are of great interest in modern epidemiological and clinical studies \citep{zeng2014efficient,zhou2014semiparametric,tao2017efficient}, because the determination of a disease's risk factor can often boil down to such a regression problem. Example \ref{eg: ATE} is of practical importance in observational studies \citep{yang2019combining}. Previous works, e.g., \cite{lin2014adjustment} and \cite{yang2019combining},  focus on the estimation in Example \ref{eg: ATE} without exploring the sampling rule design. We contribute by establishing the optimal sampling rule for a wide range of problems including Example \ref{eg: ATE}.
		
		\subsection{Example for the issue with a multi-dimensional parameter}\label{subsec: multidim}
		\begin{example}\label{ex: classification}
			Suppose $Y \in \{0,1\}$ is an indicator of some disease status and $X \in \{0,1\}$ is the test result of some fallible test for disease status. Suppose $\bV=X$ and $\bU=Y$. The prevalence of the disease $\theta_{01} = P(Y=1)$, sensitivity $\theta_{02} = P(X=1\mid Y=1)$ and specificity $\theta_{03} = P(X=0\mid Y=0)$ of the test are often of primary interest in epidemiological studies. Let $\theta_{0} = (\theta_{01}, \theta_{02}, \theta_{03})^{\T}$ be the parameter of interest. It is not hard to show the efficient influence functions of $\theta_{01}$, $\theta_{02}$ and $\theta_{03}$ are $Y - \theta_{01}$, $\theta_{01}^{-1}(X-\theta_{02})Y$ and $(1-\theta_{01})^{-1}(1-X-\theta_{03})(1-Y)$, respectively. Let $P(X) = P(Y=1\mid X)$. The conditional variances $\sigma_{1}^{2}(\bV)$, $\sigma_{2}^{2}(\bV)$ and $\sigma_{3}^{2}(\bV)$ are $P(X)(1-P(X))$, $\theta_{01}^{-2}P(X)(1-P(X))(X-\theta_{02})^{2}$ and $(1-\theta_{01})^{-2}P(X)(1-P(X))(1 - X-\theta_{03})^{2}$, which are different from each other. According to Theorem \ref{thm: optimal probability}, the optimal sampling rule for $\theta_{0j}$ is determined by $\sigma_{j}^{2}(\cdot)$. This implies that the optimal sampling rules for different parameters are different from each other. Hence, there is no sampling rule that minimizes the semiparametric efficiency bound for different parameters simultaneously in general.
			
			Suppose $\theta_{01} = 0.2$, $\theta_{02} = 0.8$, $\theta_{03} = 0.6$ and $\varpi = 0.3$. Then some numerical calculations can show that the semiparametric efficiency bound for $\theta_{03}$ under $\rho_{\rm sum}(\cdot; \tau_{\rm sum})$ and the optimal sampling rule for $\theta_{01}$ are approximately $0.35$ and $0.43$, which are both larger than that under the uniform rule ($\approx 0.30$).
		\end{example}
		
		\section{Technical Details}\label{sec: technical}
		\subsection{Regularity Conditions}\label{subsec: model settings}
		Let $F_{0}$ be the distribution of $(\bV,\bU)$. We consider the case where the parameter of interest is a general functional of $F_{0}$. Throughout this paper, we assume $\rho(\cdot)$ is bounded away from zero and $E[\|\psi\|^{2}] < \infty$ where $\|\cdot\|$ denotes the Euclidean norm.
		
		As in \cite{newey1994asymptotic}, we consider inference of a \emph{pathwise differentiable} parameter within a \emph{locally nonparametric} distribution class. Here we briefly review the definitions of ``pathwise differentiable" and ``locally nonparametric". See \cite{bickel1982adaptive,van2012unified,tsiatis2007semiparametric} for more background on semiparametric theory. 
		Let $\cF$ be a set of joint distributions of $(\bV,\bU)$ whose specific definition depends on the problem we consider. Suppose $F_{0} \in \cF$. A class of distributions $\{F_{t} : t\in [-1, 1]\}$ is called a  submodel of $\cF$ if it is contained in $\cF$ and the distribution $F_{t}$ equals to $F_{0}$ when $t = 0$. Suppose $F_{t}$ has a density $f_{t}(v,u)$ and let $S(v,u) = d \log f_{t}(v,u) / dt \big|_{t = 0}$ be the score function under the submodel. Suppose the parameter $\theta_{0} = \theta(F_{0})$ is a functional of $F_{0}$ where $\theta(\cdot)$ is a functional defined on $\cF$. Then the parameter is pathwise differentiable if there is some function $\phi(\bV,\bU)$  with zero mean and finite variance such that $d\theta(F_{t})/dt \big |_{t = 0} = E[\phi(\bV,\bU)S(\bV,\bU)]$ for any regular submodel.
		
		Pathwise differentiability is a commonly used regularity condition in semiparametric theory \citep{bickel1982adaptive}. Here, a regular submodel is a submodel that satisfies certain regularity conditions. See \cite{bickel1982adaptive} for more discussions and the formal definition of a regular submodel. Typical examples of pathwise differentiable parameters including the mean or quantile of a variable, the solution of many commonly used estimating equations among lots of other parameters.
		
		``Locally nonparametric" is a property of the distribution class $\cF$. Because all the submodels are required to belong to $\cF$, the fewer the restrictions on $\cF$, the more submodels, and hence the larger the set of score functions.
		Here, ``locally nonparametric" requires $\cF$ to be ``general" or ``unrestricted" in the sense that the set of score functions can approximate any function of $(\bV,\bU)$ with zero mean and finite variance. In a locally nonparametric distribution class, general misspecification is allowed and few restrictions are imposed except for regularity conditions \citep{newey1994asymptotic}. For example, the  distribution class which consists of all the distributions with a finite second moment is a locally nonparametric distribution class. For a missing data problem, all the observation distributions with response missing at random also consists of a locally nonparametric class.

		\subsection{Proof of Lemma \ref{lem: EIF}}
		This lemma can be obtained utilizing the techniques in the semiparametric theory for missing data problems \citep{tsiatis2007semiparametric}. To be self-contained, we provide its proof here.
		\begin{proof}
			We show the efficient influence function is
			\[h = \frac{R\psi}{\rho(\bV)} -\left(\frac{R}{\rho(\bV)} - 1\right)\Pi(\bV)\]
			and the semiparametric efficiency bound follows by straightforward calculation. The observed likelihood of $(U,V,R)$ is
			\[
			\begin{aligned}
				&f(u\mid v)^{r} f(v)\rho(v)^{r} (1 - \rho(v))^{1-r},
			\end{aligned}
			\]
			where $f(v)$ is the density of $\bV$ and $f(u\mid v)$ is the distribution of $\bU$ conditional on $\bV = v$. For any regular submodel $f_{t}(u\mid v)f_{t}(v)\rho(v)^{r} (1 - \rho(v))^{1-r}$ whose distribution is denoted by $F_{t}$, the score function is
			\begin{equation}\label{eq: score function}
				rS(u\mid v) + S(v),
			\end{equation}
			where
			\[
			\begin{aligned}
				&S(u\mid v) = \frac{d}{dt}\log f_{t}(u\mid v), \\
				&\text{and}\\
				&S(v) = \frac{d}{dt}\log f_{t}(v).\\
			\end{aligned}
			\]
			We do not consider a submodel for $\rho(v)$ since the sampling rule is determined by the researcher and hence is known in this problem. 
			Because $\psi$ is the full data influence function and $E[S(U\mid V)\mid V] = 0$, we have
			\begin{equation}\label{eq: efficiency representation}
				\begin{aligned}
					\frac{d\theta(F_{t})}{dt} 
					& = E[\psi S(\bU\mid \bV)] + E[\psi S(\bV))] \\
					& = E\left[hRS(\bU\mid \bV)\right] + E\left[hS(\bV)\right]\\
					& = E\left[h\{RS(\bU\mid \bV) + S(\bV)\}\right].
				\end{aligned}
			\end{equation}
			According to \eqref{eq: score function}, the tangent space under the two-phase design consists of all functions of the form $rS(u\mid v) + S(v)$, where $S(u\mid v)$ and $S(v)$ are the score function of $f(u\mid v)$ and $f(v)$ under some full data submodel.
			Since the full data model is locally nonparametric, the closure of the tangent space under the two-phase design consists of all score functions of the form \eqref{eq: score function}, which is
			\[
			\cT = \{rs(u,v) + s(v)\ :\  E[s(\bU,\bV)\mid \bV] = 0, \ E[s(\bV)] = 0\}.
			\]
			It is easy to verify that $h$ belongs to $\cT$. This and \eqref{eq: efficiency representation} implies $h$ is the efficient influence function according to the characterization of the efficient influence function which can be found behind Lemma 25.14 in \cite{vanderVaart2000AS}. 
		\end{proof}
		
		\subsection{Proof of Theorem \ref{thm: optimal probability}}\label{subsec: proof opt prob scalar}
		\begin{proof}
			Recall that $\rho_{\cS}(\cdot) = \rho(\cdot; \sigma, \tau_{\cS})$.
			By the definition of $\tau_{\cS}$, the sampling rule $\rho_{\cS}(\cdot)$ satisfies the constraint $E[\rho_{\cS}(\bV)] = E[\rho(\bV; \sigma, \tau_{\cS})] \leq \varpi$.
			Because the second term in the efficiency bound \eqref{eq: bound scalar} is irrelevant to the sampling rule, to show $\rho(\cdot; \sigma, \tau_{\cS})$ is the optimal sampling rule, it suffices to prove 
			\[E\left[\frac{\sigma^{2}(\bV)}{\rho^{\star}(\bV)}\right] \geq E\left[\frac{\sigma^{2}(\bV)}{\rho(\bV; \sigma, \tau_{\cS})}\right]\]
			for any sampling rule $\rho^{\star}(\cdot)$ satisfying $E[\rho^{\star}(\bV)] \leq \varpi$.
			Note that
			\[
			\begin{aligned}
				E\left[\frac{\sigma^{2}(\bV)}{\rho^{\star}(\bV)}\right] - E\left[\frac{\sigma^{2}(\bV)}{\rho(\bV; \sigma, \tau_{\cS})}\right]
				&\geq E\left[\frac{\sigma^{2}(\bV)}{\rho^{2}(\bV; \sigma, \tau_{\cS})}\left(\rho(\bV; \sigma, \tau_{\cS}) - \rho^{\star}(\bV)\right)\right]\\
				&= E\left[\sigma^{2}(\bV)\left(\rho(\bV; \sigma, \tau_{\cS}) - \rho^{\star}(\bV)\right)1\{\sigma(\bV) > \tau_{\cS}\}\right]\\
				&\quad + \tau_{\cS}^{2}E\left[\left(\rho(\bV; \sigma, \tau_{\cS}) - \rho^{\star}(\bV)\right)1\{\sigma(\bV) \leq \tau_{\cS}\}\right]\\
				&\geq \tau_{\cS}^{2}E\left[\rho(\bV; \sigma, \tau_{\cS}) - \rho^{\star}(\bV)\right] \\
				& = \tau_{\cS}^{2}(\varpi - E\left[\rho^{\star}(\bV)\right]) \geq 0,
			\end{aligned}
			\]
			where the first inequality is because $1/z_{1} - 1/z_{2} \geq (z_{2} - z_{1}) / z_{2}^{2}$ for any $z_{1}$, $z_{2} > 0$. This completes the proof.
		\end{proof}
		
		\subsection{Proof of Theorem \ref{thm: finite-dimensional equiv}}\label{subsec: proof finite equiv}
		\begin{proof}
			Recall that problem \eqref{eq: maximin} is
			\begin{equation*}
				\max_{\rho \in \cP_{\cG}} \min_{j=1,\dots,d} \left\{b_{j}^{-1}\left(\xi_{j} - E\left[\frac{\sigma_{j}^{2}(\bV)}{\rho(\bV)}\right]\right)\right\}.
			\end{equation*}
			By Lemma 1.15 in \citet{rigollet2015high}, for any $\rho(\cdot)$, we have
			\[
			\begin{aligned}
				&\min_{j=1,\dots,d} \left\{b_{j}^{-1}\left(\xi_{j} - E\left[\frac{\sigma_{j}^{2}(\bV)}{\rho(\bV)}\right]\right)\right\}\\
				& = \min_{w\in \cW^{\dag}} \left\{\sum_{j=1}^{d}w_{j}b_{j}^{-1}\xi_{j} - E\left[\frac{\sum_{j=1}^{d}w_{j}b_{j}^{-1}\sigma_{j}^{2}(\bV)}{\rho(\bV)}\right]\right\},
			\end{aligned}
			\]
			where $\cW^{\dag}=\{w=(w_{1}, \dots, w_{d}) :  \sum_{j=1}^{d}w_{j} = 1, \ 0\leq w_{j} \leq 1, \ \text{for} \ j=1,\dots,d\}$.
			Hence \eqref{eq: maximin} is equivalent to
			\begin{equation}\label{eq: equivalent problem}
				\max_{\rho \in \cP_{\cG}}\min_{w\in \cW^{\dag}} \left\{\sum_{j=1}^{d}w_{j}b_{j}^{-1}\xi_{j} - E\left[\frac{\sum_{j=1}^{d}w_{j}b_{j}^{-1}\sigma_{j}^{2}(\bV)}{\rho(\bV)}\right]\right\}.
			\end{equation}
			
			Recall that $\cP_{\cG} \colonequals \{\rho(\cdot): 0\leq \rho(\cdot)\leq 1, \ E[\rho(\bV)] \leq \varpi\}$ and $\cW^{\dag}=\{w=(w_{1}, \dots, w_{d}) :  \sum_{j=1}^{d}w_{j} = 1, \ 0\leq w_{j} \leq 1, \ \text{for} \ j=1,\dots,d\}$. Let 
			\[h(\rho,w) = \sum_{j=1}^{d}w_{j}b_{j}^{-1}\xi_{j} - E\left[\frac{\sum_{j=1}^{d}w_{j}b_{j}^{-1}\sigma_{j}^{2}(\bV)}{\rho(\bV)}\right].\]
			Take the $L_{2}$ norm and the Euclidean norm as the norm in $\cP$ and $\cW^{\dag}$, respectively.
			Then, $\cP$, $\cW^{\dag}$ are compact and $h(\rho,w)$ is continuous with respect to $\rho$ and $q$. Moreover, $h(\rho,w)$ is convex with respect to $\rho$ and linear (hence concave) w.r.t. $w$. Thus, the solution of the optimization problem does not change if we change the order of $\max$ and $\min$ in \eqref{eq: maximin} according to Theorem 3.4 in \cite{sion1958general}. Thus, the dual problem
			\begin{equation}\label{eq: minimax}
				\min_{w\in \cW^{\dag}}\max_{\rho \in \cP_{\cG}} \left\{\sum_{j=1}^{d}w_{j}b_{j}^{-1}\xi_{j} - E\left[\frac{\sum_{j=1}^{d}w_{j}b_{j}^{-1}\sigma_{j}^{2}(\bV)}{\rho(\bV)}\right]\right\}
			\end{equation}
			shares the same solution as \eqref{eq: equivalent problem}, which also leads to an equivalent problem of \eqref{eq: maximin}.
			
			According to the above derivations, we can focus on the problem \eqref{eq: minimax}.	
			Notice that the inner optimization problem of \eqref{eq: minimax}
			\[
			\begin{aligned}
				&\max_{\rho \in \cP_{\cG}} \left\{\sum_{j=1}^{d}w_{j}b_{j}^{-1}\xi_{j} - E\left[\frac{\sum_{j=1}^{d}w_{j}b_{j}^{-1}\sigma_{j}^{2}(\bV)}{\rho(\bV)}\right]\right\}\\
				& = \sum_{j=1}^{d}w_{j}b_{j}^{-1}\xi_{j} - \min_{\rho \in \cP_{\cG}}E\left[\frac{\sum_{j=1}^{d}w_{j}b_{j}^{-1}\sigma_{j}^{2}(\bV)}{\rho(\bV)}\right].
			\end{aligned}
			\]
			Similar arguments to those in the proof of Theorem \ref{thm: optimal probability} can show that $\rho(\cdot;\sigma_{w}, \tau_{w})$ minimizes the functional
			$E\left[\sum_{j=1}^{d}w_{j}b_{j}^{-1}\sigma_{j}^{2}(\bV)/\rho(\bV)\right]$ over $\cP_{\cG}$ and the minimum value is 
			\[
			\min_{\rho \in \cP_{\cG}}E\left[\frac{\sum_{j=1}^{d}w_{j}b_{j}^{-1}\sigma_{j}^{2}(\bV)}{\rho(\bV)}\right] = E\left[\sigma_{w}(\bV)\max\{\sigma_{w}(\bV), \tau_{w} \}\right],
			\]
			where
			$\sigma_{w}(\bV) = \sqrt{\sum_{j=1}^{d}w_{j}b_{j}^{-1}\sigma_{j}^{2}(\bV)}$
			and $\tau_{w}$  is the unique solution of $E[\rho(\bV;\sigma_{w},\tau)] = \varpi$ with respect to $\tau$. This completes the proof of Theorem \ref{thm: finite-dimensional equiv}. 
		\end{proof}  
		
		\subsection{Proof of Theorem \ref{thm: exact budget constraint}}\label{subsec: proof exact budget}
		\begin{proof}
			We prove the result for $\tilde{\rho}_{j}(\cdot)$ for $j = 1,\dots, d$. The result for $\tilde{\rho}_{\cC}(\cdot)$ and $\tilde{\rho}_{\cG}(\cdot)$ can be established similarly.
			For $i = 1,\dots, n$, the expectation of $R_{2i}$ is $(1 - R_{1i})\tilde{\rho}_{j}(\bV_{i})$
			conditional on $(R_{11}, \bV_{1}),\dots,$ $(R_{1n}, \bV_{n})$ and $U_{j}$ for $j$ with $R_{1j} = 1$. Thus conditional on the same variables, the expectation of $\sum_{i=1}^{n}(R_{1i} + R_{2i})$ is $\sum_{i=1}^{n}R_{1i} + \sum_{i=1}^{n}(1 - R_{1i})\tilde{\rho}_{j}(\bV_{i})$. Because $\tilde{\tau}_{j}$ is the solution of \eqref{eq: threshold estimating equation}, we have $\sum_{i=1}^{n}(1 - R_{1i})\tilde{\rho}_{j}(\bV_{i}) = (\varpi - \kappa_{n})n$. According to the law of iterated conditional expectation, we have $E\left[\sum_{i=1}^{n}(R_{1i} + R_{2i})\right] = \kappa_{n} n + (\varpi - \kappa_{n})n = \varpi n$ which proves Theorem \ref{thm: exact budget constraint}.
		\end{proof}
		
		\subsection{Proof of Theorem \ref{thm: convergence rule comp}}\label{subsec: proof of convergence rule comp}
		In this and the following proofs, we use $M$ to denote generic positive constants whose values may be different in different places. 
		We first get down to the required regularity conditions. Recall that $\tau_{j}$ is the solution of $E[\rho(V; \sigma_{j}, \tau)] = \varpi$ for $j = 1,\dots, d$. 
		\begin{condition}\label{cond: id comp}
			There is some constants $r_{j},L_{j} >0$ such that $r_{j} < \tau_{j}$ and $|E[\rho(\bV; \sigma_{j}, \tau_{1})] - E[\rho(\bV;\sigma_{j}, \tau_{2})]| > L_{j} |\tau_{1} - \tau_{2}|$ for any $\tau_{1}, \tau_{2} \in [\tau_{j} - r_{j}, \tau_{j} + r_{j}]$ and $j=1,\dots,d$, where $\rho(\cdot; \sigma_{j}, \tau) = 1\{\sigma_{j}(\bV) > \tau\} + \sigma_{j}(\bV)/\tau1\{\sigma_{j}(\bV) \leq \tau\}$.
		\end{condition}
		\begin{condition}\label{cond: bound cond mean-var}
			$\sup_{v}\Pi_{j}(v) < \infty$ and
			$0 < \inf_{v}\sigma_{j}(v) \leq \sup_{v} \sigma_{j}(v) < \infty$ for $j = 1,\dots,d$.
		\end{condition}
		Condition \ref{cond: id comp} requires that the budgets under different thresholds are different in a neighborhood of $\tau_{j}$. Condition \ref{cond: bound cond mean-var} is a mild regularity condition. Next, we give the proof of Theorem \ref{thm: convergence rule comp}.
		\begin{proof}
			We prove the results for $\tilde{\rho}_{j}(\cdot)$ for $j = 1,\dots,d$. The result for $\tilde{\rho}_{\cS}(\cdot)$ is a special case of $d = 1$.
			
			We first show $\tilde{\tau}_{j}$ converges to $\tau_{j}$ for $j = 1, \dots, d$, where $\tilde{\tau}_{j}$ is the solution of Equation \eqref{eq: threshold estimating equation} in the main text. Let $\tau_{{j},n}$ be the solution of $E[\rho(\bV; \sigma_{j}, \tau)] = (\varpi - \kappa_{n}) / (1 - \kappa_{n})$. Note that $(\varpi - \kappa_{n})/(1 - \kappa_{n}) - \varpi = O(\kappa_{n})$. Under Condition \ref{cond: id comp}, we have $|\tau_{{j},n} - \tau_{j}| = O(\kappa_{n})$. Next, we show that $|\tilde{\tau}_{j} - \tau_{{j}, n}|$ converges to zero.
			For $\tau \in [\tau_{j} - r_{j}, \tau_{j} + r_{j}]$, define
			\[h_{{j}, n}(\tau) = \frac{1}{n}\sum_{i=1}^{n}(1 - R_{1i})\left(1\{\sigma(\bV_{i})\geq \tau\} + \frac{\sigma_{j}(\bV_{i})}{\tau}1\{\sigma_{j}(\bV_{i}) < \tau\}\right)\]
			and 
			\[\tilde{h}_{{j}, n}(\tau) = \frac{1}{n}\sum_{i=1}^{n}(1 - R_{1i})\left(1\{\tilde{\sigma}_{j}(\bV_{i})\geq \tau\} + \frac{\tilde{\sigma}_{j}(\bV_{i})}{\tau}1\{\tilde{\sigma}_{j}(\bV_{i}) < \tau\}\right).\]
			By calculating the mean and variance, we have
			\begin{equation}\label{eq: bound gsg 1}
				|h_{{j}, n}(\tau) - E[h_{{j}, n}(\tau)]| = O_{P}\left(\frac{1}{\sqrt{n}}\right)
			\end{equation}
			uniformly in $\tau \in [\tau_{j} - r_{j}, \tau_{j} + r_{j}]$. Moreover, we have $E[h_{{j}, n}(\tau)] = (1-\kappa_{n})E[\rho(\bV; \sigma_{j}, \tau)]$ which implies $E[h_{{j}, n}(\tau)] = \varpi - \kappa_{n}$. By Condition \ref{cond: rate sigma}, it is not hard to verify
			\begin{equation}\label{eq: bound gsg 2}
				|\tilde{h}_{{j},n}(\tau) - h_{{j}, n}(\tau)| \leq \frac{1}{\tau_{j} - r}\|\tilde{\sigma} - \sigma\|_{\infty} = O_{P}\left\{(n\kappa_{n})^{-\delta}\right\}
			\end{equation}
			uniformly in  $\tau \in [\tau_{j} - r_{j}, \tau_{j} + r_{j}]$, where $\delta$ is a constant determined by the convergence rate of $\|\tilde{\sigma}_{j} - \sigma_{j}\|_{\infty}$ which appears in Condition \ref{cond: rate sigma} . Combining \eqref{eq: bound gsg 1} and \eqref{eq: bound gsg 2}, we have
			\[|\tilde{h}_{{j},n}(\tau) - E[h_{{j}, n}(\tau)]| = O_{P}\left\{(n\kappa_{n})^{-\delta}\right\}. \]
			Thus for any $\epsilon > 0$, there is some constant $M > 0$ such that 
			\[P\left(|\tilde{h}_{{j}, n}(\tau) - E[h_{j, n}(\tau)]| \geq M(n\kappa_{n})^{-\delta}\right) \leq \frac{\epsilon}{2}\]
			for any $\tau \in [\tau_{j} - r_{j}, \tau_{j} + r_{j}]$.
			Define $\tau_{{j}, n, 1} = \tau_{{j}, n} - M(n\kappa_{n})^{-\delta}/(L_{j}(1-\kappa_{n}))$ and $\tau_{{j}, n, 2} = \tau_{{j}, n} + M(n\kappa_{n})^{-\delta}/(L_{j}(1-\kappa_{n}))$. Because $|\tau_{{j}, n} - \tau_{j}| = O(\kappa_{n})$, $\kappa_{n} \to 0$ and $n\kappa_{n} \to \infty$, we have $\tau_{{j}, n}$, $\tau_{{j}, n, 1}$, $\tau_{{j}, n, 2} \in [\tau_{j} - r_{j}, \tau_{j} + r_{j}]$ for sufficiently large $n$. Hence
			\begin{equation}\label{eq: bound gsg 3}
				\begin{aligned}
					&P\Big(\tilde{h}_{{j}, n}(\tau_{{j},n,1}) - E[h_{{j}, n}(\tau_{{j},n,1})] < M(n\kappa_{n})^{-\delta},
					\tilde{h}_{{j}, n}(\tau_{{j},n,2}) - E[h_{{j}, n}(\tau_{{j},n,2})] > - M(n\kappa_{n})^{-\delta}\Big)\\
					&\geq
					1 - P\left(\tilde{h}_{{j}, n}(\tau_{{j},n,1}) - E[h_{{j}, n}(\tau_{{j},n,1})] \geq M(n\kappa_{n})^{-\delta}\right) \\
					&\quad - P\left(\tilde{h}_{{j}, n}(\tau_{{j},n,2}) - E[h_{{j}, n}(\tau_{{j},n,2})] \leq - M(n\kappa_{n})^{-\delta}\right)\\
					& \geq 1 - \epsilon
				\end{aligned}
			\end{equation}
			for sufficiently large $n$. According to Condition \ref{cond: id comp} and the monotonicity of $E[h_{{j}, n}(\tau)] = (1-\kappa_{n})E[\rho(\bV; \sigma_{j}, \tau)]$, we have 
			\[E[h_{{j}, n}(\tau_{{j},n,1})] \leq \varpi - \kappa_{n} - M(n\kappa_{n})^{-\delta}\]
			and 
			\[E[h_{{j}, n}(\tau_{{j},n,2})] \geq \varpi - \kappa_{n} + M(n\kappa_{n})^{-\delta}.\]
			Combining this with \eqref{eq: bound gsg 3}, we have
			\[P\left(\tilde{h}_{{j}, n}(\tau_{{j},n,1}) < \varpi - \kappa_{n} < \tilde{h}_{{j}, n}(\tau_{{j},n,2})\right) \geq 1 - \epsilon.\]
			Recall that $\tilde{\tau}_{j}$ is the solution of $\tilde{h}_{j, n}(\tau) = \varpi - \kappa_{n}$.
			We have 
			$P\left(\tilde{\tau}_{j} \in [\tau_{{j},n,1}, \tau_{{j},n,2}]\right) \geq 1 - \epsilon$ by the monotonicity of $\tilde{h}_{{j}, n}(\tau)$.
			This implies $|\tilde{\tau}_{j} - \tau_{{j}, n}| = O_{P}\{(n\kappa_{n})^{-\delta}\}$ and hence 
			\begin{equation}\label{eq: converge tau tilde}
				|\tilde{\tau}_{j} - \tau_{j}| = O_{P}\{(n\kappa_{n})^{-\delta} + \kappa_{n}\}.
			\end{equation}
			
			Under Condition \ref{cond: bound cond mean-var}, we have
			\[
			\begin{aligned}
				&|\tilde{\rho}_{j}(v) - \rho_{j}(v)|\\ 
				& = \left| 1\{\tilde{\sigma}_{j}(v) > \tilde{\tau}_{j}\} + \frac{\tilde{\sigma}_{j}(v)}{\tilde{\tau}_{j}}1\{\tilde{\sigma}_{j}(v) \leq \tilde{\tau}_{j}\}  - 1\{\sigma_{j}(v) > \tau_{j}\} - \frac{\sigma_{j}(v)}{\tau_{j}}1\{\sigma_{j}(v) \leq \tau_{j}\}\right|\\
				&\leq \left|1\{\tilde{\sigma}_{j}(v) > \tilde{\tau}_{j}\} + \frac{\tilde{\sigma}_{j}(v)}{\tilde{\tau}_{j}}1\{\tilde{\sigma}_{j}(v) \leq \tilde{\tau}_{j}\}  - 1\{\sigma_{j}(v) > \tilde{\tau}_{j}\} - \frac{\sigma_{j}(v)}{\tilde{\tau}_{j}}1\{\sigma_{j}(v) \leq \tilde{\tau}_{j}\}\right| \\
				& + \left| 1\{\sigma_{j}(v) > \tilde{\tau}_{j}\} + \frac{\sigma_{j}(v)}{\tilde{\tau}_{j}}1\{\sigma_{j}(v) \leq \tilde{\tau}_{j}\}  - 1\{\sigma_{j}(v) > \tau_{j}\} - \frac{\sigma_{j}(v)}{\tau_{j}}1\{\sigma_{j}(v) \leq \tau_{j}\}\right|\\
				& \leq M (\|\tilde{\sigma}_{j} - \sigma_{j}\|_{\infty} + |\tilde{\tau}_{j} - \tau_{j}|)
			\end{aligned}
			\]
			for some constant $M$ with probability approaching one. Combining this with Condition \ref{cond: rate sigma} and Equation \eqref{eq: converge tau tilde}, we have
			\[\|\tilde{\rho}_{j} - \rho_{j}\|_{\infty} = O_{P}\left\{(n\kappa_{n})^{-\delta} + \kappa_{n}\right\},\]
			which completes the proof.
		\end{proof}
		
		\subsection{Proof of Theorem \ref{thm: convergence rule mv}}\label{subsec: proof of thm rule mv}
		In this subsection, we turn to the convergence result of the estimated sampling rule in the multi-dimensional parameter case.
		For $w \in \cW$, define
		\[
		H_{j}(w) = - b_{j}^{-1}\left(\xi_{j} - E\left[\frac{\sigma_{j}^{2}(\bV)}{\rho_{0}(\bV) + \sum_{j=1}^{d}w_{j}(\rho_{j}(\bV) - \rho_{0}(\bV_{i}))}\right]\right)
		\]
		and
		\[
		H_{\cC}(w) = \max_{j=1,\dots,d} H_{j}(w),
		\]
		where $\xi_{j} = E[\sigma_{j}^{2}(V) / \rho_{0}(V)]$ and $b_{j} = E\left[\sigma_{j}^{2}(\bV)/\rho_{0}(\bV)\right] + \Var[\Pi_{j}(\bV)]$ for $j = 1,\dots, d$.
		Similarly, for $w\in \cW^{\star}$, let 
		\[
		H_{\cG}(w) = \sum_{j=1}^{d}w_{j}b_{j}^{-1}\xi_{j} - E\left[\sigma_{w}(\bV)\max\{\sigma_{w}(\bV), \tau_{w} \}\right].
		\]
		Then, $w_{\cC} = \mathop{\arg\min}_{w \in \cW}H_{\cC}(w)$ and $w_{\cG} = \mathop{\arg\min}_{w \in \cW^{\dag}}H_{\cG}(w)$. The following condition is required to establish the convergence rate of $\tilde{\rho}_{\cC}(\cdot)$ and $\tilde{\rho}_{\cG}(\cdot)$.
		\begin{condition}\label{cond: benchmark rule}
			The benchmark sampling rule $\rho_{0}(\cdot)$ is bounded away from zero and satisfies $E[\rho_{0}(\bV)] = \varpi$.
		\end{condition}
		
		\begin{condition}\label{cond: local convexity Gc}
			(i) $H_{\cC}(w)$ has the unique minimum point $w_{\cC}$; (ii) for some constants $r_{\cC}, L_{\cC} > 0$  and any $w\in \cW$ such that $\|w - w_{\cC}\| \leq r_{\cC}$, we have $H_{\cC}(w) - H_{\cC}(w_{\cC}) \geq L_{\cC}\|w - w_{\cC}\|^{2}$.
		\end{condition}
		
		\begin{condition}\label{cond: local convexity Gm}
			(i) $H_{\cG}(w)$ has the unique minimum point $w_{\cG}$; (ii) for some constants $r_{\cG}, L_{\cG} > 0$ and any $w\in \cW^{\dag}$ such that $\|w - w_{\cG}\| \leq r_{\cG}$, we have $H_{\cG}(w) - H_{\cG}(w_{\cG}) \geq L_{\cG}\|w - w_{\cG}\|^{2}$.
		\end{condition}
		Condition \ref{cond: benchmark rule} is a regularity condition on the benchmark sampling rule $\rho_{0}(\cdot)$. Condition \ref{cond: local convexity Gc} is a mild regularity condition, which can be satisfied if $H_{\cC}(w)$ has a continuous Hessian matrix in a neighborhood of $r_{\cC}$ and the Hessian matrix at  $r_{\cC}$ is positive definite. Condition \ref{cond: local convexity Gm} is similar to Condition \ref{cond: local convexity Gc} with $H_{\cC}(w)$ replaced by $H_{\cG}(w)$.  Now, we are ready to prove Theorem \ref{thm: convergence rule mv}.
		
		\begin{proof}
			We only prove the result for $\tilde{\rho}_{\cC}(\cdot)$. The result for $\tilde{\rho}_{\cG}(\cdot)$ can be established similarly.
			
			Recall that $\rho_{\cC}(\cdot) = \rho_{0}(\cdot) + \sum_{j=1}^{d}w_{{\cC},j}\{\rho_{j}(\cdot) - \rho_{0}(\cdot)\}$ and $\tilde{\rho}_{\cC}(\cdot) = \rho_{0}(\cdot) + \sum_{j=1}^{d}\widetilde{w}_{{\cC},j}(\tilde{\rho}_{j}(\cdot) - \rho_{0}(\cdot))$, where $w_{{\cC},j}$ and $\widetilde{w}_{{\cC}, j}$ are the $j$th component of $w_{\cC}$ and $\widetilde{w}_{\cC}$, respectively, for $j=1,\dots,d$. Recall that $\rho_{j}(v) = 1\{\sigma_{j}(v) > \tau_{j}\} + \sigma_{j}(v)1\{\sigma_{j}(v) \leq \tau_{j}\}/\tau_{j}$ for $j = 1,\dots, d$. By Condition \ref{cond: bound cond mean-var}, we have 
			\begin{equation}\label{eq: bounded rule}
				\frac{1}{M} \leq \inf_{v}\rho_{j}(v) \leq \sup_{v}\rho_{j}(v) \leq M
			\end{equation}
			for some $M > 1$. Hence 
			\[
			\begin{aligned}
				\|\tilde{\rho}_{\cC} - \rho_{\cC}\|_{\infty} \leq \sum_{j = 1}^{d}\|\tilde{\rho}_{j} - \rho_{j}\|_{\infty} + M\|\tilde{w}_{\cC} - w_{\cC}\|.
			\end{aligned}
			\]
			Theorem \ref{thm: convergence rule comp} has established the convergence rate of $\|\tilde{\rho}_{j} - \rho_{j}\|_{\infty}$. Thus, in order to establish the convergence rate of $\|\tilde{\rho}_{\cC} - \rho_{\cC}\|_{\infty}$, it suffices to establish the convergence rate of $\|\widetilde{w}_{\cC} - w_{\cC}\|$ .
			Define 
			\[\]\[
			\widetilde{H}_{j}(w) = - \tilde{b}_{j}^{-1}\left(\tilde{\xi}_{j} - \frac{1}{n}\sum_{i=1}^{n}\frac{\tilde{\sigma}_{j}^{2}(\bV)}{\rho_{0}(\bV) + \sum_{j=1}^{d}w_{j}(\tilde{\rho}_{j}(\bV) - \rho_{0}(\bV_{i}))}\right)
			\]
			and
			\[
			\widetilde{H}_{\cC}(w) = \max_{j=1,\dots,d} \widetilde{H}_{j}(w),
			\] 
			where, for $j = 1,\dots, d$, $\tilde{b}_{j} = n^{-1}\sum_{i=1}^{n}\tilde{\sigma}_{j}^{2}(\bV_{i})/\rho_{0}(\bV_{i}) + n^{-1}\sum_{i=1}^{n}\{\widetilde{\Pi}_{j}(V_{i}) - n^{-1}\sum_{i=1}^{n}\widetilde{\Pi}_{j}(V_{i})\}^{2}$ and $\tilde{\xi}_{j} = n^{-1}\sum_{i=1}^{n}\tilde{\sigma}_{j}^{2}(\bV_{i})/\rho_{0}(\bV_{i})$.
			Then $\widetilde{w}_{\cC} = \mathop{\arg\min}_{w\in \cW}\widetilde{H}_{\cC}(w)$. Define $\bar{b}_{j}$ and $\bar{H}_{j}(w)$ similarly to $\tilde{b}_{j}$ and $\widetilde{H}_{j}(w)$ with $\widetilde{\Pi}_{j}(\cdot)$, $\tilde{\sigma}_{j}(\cdot)$ and $\tilde{\rho}_{j}(\cdot)$ in $\tilde{b}_{j}$ and $\widetilde{H}_{j}(w)$ replaced by $\Pi_{j}(\cdot)$, $\sigma_{j}(\cdot)$ and $\rho_{j}(\cdot)$. Let $\bar{H}_{\cC}(w) = \max_{j=1,\dots,d}\bar{H}_{j}(w)$. 
			According to Condition \ref{cond: rate sigma}, Theorem \ref{thm: convergence rule comp}, and inequality \eqref{eq: bounded rule}, there is some constant $M>1$ such that $1/M \leq \inf_{v}\tilde{\rho}_{j}(v) \leq \sup_{v}\tilde{\rho}_{j}(v) \leq M$ with probability approaching one. Conditions \ref{cond: bound cond mean-var}, \ref{cond: benchmark rule} imply that $\max_{j = 1,\dots, d}|\tilde{b}_{j} - \bar{b}_{j}| \leq  M(\|\widetilde{\Pi}_{j} - \Pi_{j}\|_{\infty} + \|\tilde{\sigma}_{j} - \sigma_{j}\|_{\infty} + \|\tilde{\rho}_{j} - \rho_{j}\|_{\infty})$ for some constant $M$.
			Then, Conditions \ref{cond: bound cond mean-var}, \ref{cond: benchmark rule} and the mean value theorem implies that
			\[\sup_{w\in\cW}|\widetilde{H}_{j}(w) - \bar{H}_{j}(w)| \leq M(\|\widetilde{\Pi}_{j} - \Pi_{j}\|_{\infty} + \|\tilde{\sigma}_{j} - \sigma_{j}\|_{\infty} + \|\tilde{\rho}_{j} - \rho_{j}\|_{\infty})\]
			with probability approaching one for some constant $M$. This implies
			\[
			\begin{aligned}
				&\sup_{w\in\cW}|\widetilde{H}_{\cC}(w) - \bar{H}_{\cC}(w)| \leq\\
				&M\left(\max_{j=1,\dots,d}\|\widetilde{\Pi}_{j} - \Pi_{j}\|_{\infty} + \max_{j=1,\dots,d}\|\tilde{\sigma}_{j} - \sigma_{j}\|_{\infty} + \max_{j=1,\dots,d}\|\tilde{\rho}_{j} - \rho_{j}\|_{\infty}\right)
			\end{aligned}
			\]
			with probability approaching one.
			Then, according to Condition \ref{cond: rate sigma} and Theorem \ref{thm: convergence rule comp}, we have
			\begin{equation}\label{eq: est sup bound Gc}
				\sup_{w\in\cW}|\widetilde{H}_{\cC}(w) - \bar{H}_{\cC}(w)| = O_{P}\left\{(n\kappa_{n})^{-\delta}+\kappa_{n}\right\}.
			\end{equation}
			Recall that
			\[
			\bar{b}_{j} = \left(\frac{1}{n}\sum_{i=1}^{n}\left[  \frac{\sigma_{j}^{2}(\bV_{i})}{\rho_{0}(\bV_{i})} + \left\{\Pi_{j}(V_{i}) - \frac{1}{n}\sum_{i=1}^{n}\Pi_{j}(V_{i}) \right\}^{2}\right]\right)
			\]
			\[
			\begin{aligned}
				\bar{H}_{j}(w)
				& = \bar{b}_{j}^{-1} \frac{1}{n}\sum_{i=1}^{n}\left\{\frac{\sigma_{j}^{2}(\bV_{i})}{\rho_{0}(\bV_{i})} - \frac{\sigma_{j}^{2}(\bV_{i})}{\rho_{0}(\bV_{i}) + \sum_{j=1}^{d}w_{j}(\rho_{j}(\bV_{i}) - \rho_{0}(\bV_{i}))}\right\}
			\end{aligned}
			\]  
			for any $w\in \cW$. According to Condition \ref{cond: bound cond mean-var}, we have $b_{j} \geq 1 / M$ for some constant $M$ and $j = 1,\dots, d$. Notice that the function $\sigma_{j}(v)/\rho_{0}(v) + \sigma_{j}(v)/[\rho_{0}(v) + \sum_{j=1}^{d}w_{j}(\rho_{j}(v) - \rho_{0}(v))]$ is bounded and Lipschitz continuous with respect to $w$ due to \eqref{eq: bounded rule}, Condition \ref{cond: bound cond mean-var}, and Condition \ref{cond: benchmark rule}. Then, 
			similar concentration arguments as in the proof of Lemma C.1 in \cite{he2021smoothed} can show that
			\begin{equation}\label{eq: concentration sup bound Gc}
				\sup_{w\in \cW}|\bar{H}_{j}(w) - H_{j}(w)| = O_{P}\left(\frac{1}{\sqrt{n}}\right).
			\end{equation}
			Combining \eqref{eq: est sup bound Gc} with \eqref{eq: concentration sup bound Gc}, we have 
			\[
			\sup_{w\in \cW}|\widetilde{H}_{j}(w) - H_{j}(w)| = O_{P}\left\{(n\kappa_{n})^{-\delta}+\kappa_{n}\right\}.
			\]
			Thus, for any $\epsilon > 0$, there is some $M > 0$ such that $P(\sup_{w\in \cW}|\widetilde{H}_{j}(w) - H_{j}(w)| \geq M\{(n\kappa_{n})^{-\delta}+\kappa_{n}\}) \leq \epsilon$. Because $\widetilde{w}_{\cC}$ is the minimum point of $\widetilde{H}_{\cC}(w)$, we have
			$\widetilde{H}_{\cC}(\widetilde{w}_{\cC}) \leq \widetilde{H}_{\cC}(w_{\cC})$ and hence
			\[
			\begin{aligned}
				H_{\cC}(\widetilde{w}_{\cC}) - H_{\cC}(w_{\cC})  
				&\leq \widetilde{H}_{\cC}(\widetilde{w}_{\cC}) - \widetilde{H}_{\cC}(w_{\cC}) + 2 \sup_{w\in \cW}|\widetilde{H}_{j}(w) - H_{j}(w)| \\
				& \leq 2\sup_{w\in \cW}|\widetilde{H}_{j}(w) - H_{j}(w)|.
			\end{aligned}
			\]
			Thus 
			\begin{equation}\label{eq: value bound}
				H_{\cC}(\widetilde{w}_{\cC}) - H_{\cC}(w_{\cC}) \leq 2M\{(n\kappa_{n})^{-\delta}+\kappa_{n}\}
			\end{equation}
			with probability at least $1 - \epsilon$. Because $H_{\cC}(w)$ is continuous with respect to $w$, $\cW$ is a compact set and $w_{\cC}$ is the unique minimum point of $H_{\cC}(w)$, we have $\inf_{w\in \cW, \|w - w_{\cC}\|\geq r_{\cC}}\{H_{\cC}(w) - H_{\cC}(w_{\cC})\} > 0$ for $r_{\cC}$ in Condition \ref{cond: local convexity Gc}. For sufficiently large $n$ such that $2M\{(n\kappa_{n})^{-\delta}+\kappa_{n}\} < \inf_{w\in \cW, \|w - w_{\cC}\|\geq r_{\cC}}\{H_{\cC}(w) - H_{\cC}(w_{\cC})\}$, we have $\|\widetilde{w}_{\cC} - w_{\cC}\| \leq r_{\cC}$ with probability at least $1 - \epsilon$ according to \eqref{eq: value bound}. Then inequality \eqref{eq: value bound} and Condition \ref{cond: local convexity Gc} imply $\|\widetilde{w}_{\cC} - w_{\cC}\| \leq \sqrt{2M\{(n\kappa_{n})^{-\delta}+\kappa_{n}\}/L_{\cC}}$ with probability at least $1 - \epsilon$. Because $\epsilon$ is arbitrary, we have $\|\widetilde{w}_{\cC} - w_{\cC}\| = O_{P}(\sqrt{(n\kappa_{n})^{-\delta}+\kappa_{n}})$. Combining this with Condition \ref{cond: bound cond mean-var} and the fact that $\|\tilde{\rho}_{j} - \rho_{j}\|_{\infty} = O_{P}\{(n\kappa_{n})^{-\delta}+\kappa_{n}\}$, we conclude that $\|\tilde{\rho}_{\cC} - \rho_{\cC}\|_{\infty} = O_{P}(\sqrt{(n\kappa_{n})^{-\delta}+\kappa_{n}})$.
		\end{proof}

		\section{Estimation}\label{sec: est}
		\subsection{Estimate the Conditional Mean and Variance of the Efficient Influence Function}\label{subsec: estimate EIF}
		The full data efficient influence function depends on $\theta_{0}$ and may also depend on some unknown nuisance parameters, e.g., $\alpha_{0}$ and $\beta_{0}$ in Example \ref{eg: least squares} and $m_{1}(\cdot)$, $m_{0}(\cdot)$ and $\pi(\cdot)$ in Example \ref{eg: ATE}. Thus, we need to estimate these unknown quantities. We write the efficient influence function as $\psi(V, \bU; \theta_{0}, \eta_{0})$ where $\eta_{0}$ is the nuisance parameter. 
		The nuisance parameter $\eta_{0}$ can be estimated using the pilot sample. Denote the resulting estimator by $\tilde{\eta}$. Then $\theta_{0}$ can be estimated by $\tilde{\theta}$ which is the solution of the estimating equation
		\[\sum_{i=1}^{n}R_{1i}\psi(\bV_{i},\bU_{i}; \theta, \tilde{\eta}) = 0,\]
		where $R_{1i}$ is the inclusion indicator for the pilot sample.
		Then we obtain the estimates $\tilde{\psi}_{i} = \psi(\bV_{i}, \bU_{i}; \tilde{\theta}, \tilde{\eta})$ ($i: R_{1i} = 1$) for the full data efficient influence function of observations in the pilot sample.
		
		In order to estimate the optimal sampling rule and construct efficient estimator according to the efficient influence function in Lemma \ref{lem: EIF}, one needs to estimate the 
		conditional mean $\Pi_{j}(\cdot)$ and the standard deviation $\sigma_{j}(\cdot)$ for $j = 1,\dots, d$. These quantities can be estimated by fitting a heteroscedastic parametric regression model using the pilot sample and the estimated $\tilde{\psi}_{i}$'s. 
		In practice, it may be hard to model $\Pi_{j}(\cdot)$ and $\sigma_{j}(\cdot)$  for $j = 1,\dots, d$. If plausible parametric models are not available, we recommend to estimate them nonparametrically. Many methods are available for this task, including kernel smoothing \citep{fan1998efficient,fan2018local}, sieve methods \citep{huang1998projection,huang2001concave,chen2007large} and kernel ridge regression \citep{mendelson2002geometric}. A sieve method that can estimate the conditional mean and standard deviation simultaneously is proposed in the next section. The proposed method is computationally efficient and performs well even when the dimension of the first-phase variable is moderately high.
		
		\subsection{New Nonparametric Estimators for the Conditional Mean and Variance} \label{subsec: MVR}
		If a plausible model for the conditional mean function $\Pi_{j}(v)$ is unavailable for $j = 1,\dots, d$, we can
		approximate it with a linear combination of some basis functions such as polynomials, wavelets, or splines. Let $p(v) = (p_{1}(v), \dots, p_{K}(v))^{\T}$ be a vector of basis functions which can change with $n$. One can increase $K$ with $n$ to make the approximation more and more accurate as the sample size increases. Then, we approximate $\Pi_{j}(v)$ by $\gamma_{1}^{\T}p(v)$ for some $\gamma_{1}$. The conditional standard deviation $\sigma_{j}(v)$ can be approximate similarly by $\gamma_{2}^{\T}p(v)$ for some $\gamma_{2}$. However, this approximation is not guaranteed to be non-negative. An infeasible negative sampling rule is obtained if we plug the negative approximation into the expression of the optimal sampling rule in Theorem \ref{thm: optimal probability}. Using the truncated version $\max\{\gamma_{2}^{\T}p(v),0\}$ can avoid this problem but leads the function not differentiable with respect to $\gamma_{2}$ which makes optimization difficult. Hence we propose to use a transformation function $\Lambda(v) = \log(1 + \exp(v))$ and use $\Lambda(\gamma_{2}^{\T}p(v))$ to approximate $\sigma_{j}(v)$. The function $\Lambda(v)$ is a smooth approximation of the truncation function $\max\{v,0\}$. Moreover, it is convex, Lipchitz continuous, differentiable, and non-negative. These nice properties benefit the optimization.
		
		Then, the remaining problem is to determine $\gamma_{1}$ and $\gamma_{2}$. The discussion behind Theorem \ref{thm: optimal probability} can show that $\sigma_{j}(\cdot)$ minimizes $E[(\psi_{j} - \Pi_{j}(\bV))^{2}/f_{2}(\bV)]$ over all positive $f_{2}(\cdot)$ such that $E[f_{2}(\bV)] \leq E[\sigma_{j}(\bV)]$. This motivates us to consider the penalized objective formulation of the constrained optimization problem $E[(\psi_{j} - \Pi_{j}(\bV))^{2}/f_{2}(\bV)] + E[f_{2}(\bV)]$. One can verify that this objective function is minimized if $f_{2}(\cdot) = \sigma_{j}(\cdot)$. On the other hand, it is straightforward to show, for any given positive function $f_{2}(\cdot)$, the conditional mean function $\Pi_{j}(\cdot)$ minimizes the weighted least squares objective function $E[(\psi_{j} - f_{1}(\bV))^{2} / f_{2}(\bV)] + E[f_{2}(\bV)]$ over all $f_{1}(\cdot)$. This inspires us to recover the conditional mean and variance simultaneously by minimizing $E[(\psi_{j} - f_{1}(\bV))^{2} / f_{2}(\bV)] + E[f_{2}(\bV)]$ with respect to $f_{1}(\cdot)$ and $f_{2}(\cdot)$. By replacing the expectation with sample mean and plugging in the estimates and approximations, we obtain the objective function
		\begin{equation}\label{eq: MVR}
			\cL_{nj}(\gamma_{1},\gamma_{2}) = \frac{1}{n\kappa_{n}}\sum_{i=1}^{n}R_{1i}\left\{\frac{(\tilde{\psi}_{ij} - \gamma_{1}^{\T}p(V_{i}))^{2}}{\Lambda(\gamma_{2}^{\T}p(\bV_{i}))} + \Lambda(\gamma_{2}^{\T}p(\bV_{i}))\right\},
		\end{equation} 
		where $\tilde{\psi}_{ij}$ is the $j$th component of $\tilde{\psi}_{i}$ for $j = 1,\dots, d$ and $i = 1,\dots, n$.
		Let $\tilde{\gamma}_{1j},$ $\tilde{\gamma}_{2j}$ be the minimum point of \eqref{eq: MVR}. Then $\Pi_{j}(\bV)$ and $\sigma_{j}(\bV)$ can be estimated by $\widetilde{\Pi}_{j}(\bV) = \tilde{\gamma}_{1j}^{\T}p(\bV)$ and $\tilde{\sigma}_{j}(\bV) = \Lambda(\tilde{\gamma}_{2j}^{\T}p(\bV))$, respectively. 
		The proposed objective function has the following block-wise convex property.
		\begin{proposition}\label{prop: optimization MVR}
			For $j = 1,\dots, d$ and any give $\gamma_{1}$, $\cL_{nj}(\gamma_{1},\gamma_{2})$ is convex with respect to $\gamma_{2}$; for $j = 1,\dots, d$ and  any give $\gamma_{2}$, $\cL_{nj}(\gamma_{1},\gamma_{2})$ is convex with respect to $\gamma_{1}$.
		\end{proposition}
		\begin{proof}
			For any given $\gamma_{1}$, the Hessian of $\cL_{nj}(\gamma_{1},\gamma_{2})$ with respect to $\gamma_{2}$ is
			\[
			\begin{aligned}
				&\frac{1}{n\kappa_{n}}\sum_{i=1}^{n}R_{1i}p(\bV_{i})p(\bV_{i})^{\T}
				\Big\{(\tilde{\psi}_{i} - \gamma_{1}^{\T}p(\bV_{i}))^{2}\times\\
				&\left[2\Lambda(\gamma_{2}^{\T}p(\bV_{i}))^{-3}\Lambda^{(1)}(\bV_{i})^{2} - \Lambda(\gamma_{2}^{\T}p(\bV_{i}))^{-2}\Lambda^{(2)}(\bV_{i})^{2} \right] + \Lambda^{(2)}(\gamma_{2}^{\T}p(\bV_{i}))\Big\}
			\end{aligned}
			\]
			where $\Lambda^{(1)}(v) = \exp(v)/\{1+\exp(v)\}$ and $\Lambda^{(2)}(v) = \exp(v)/\{1+\exp(v)\}^2$. This matrix is positive semi-definite because $\Lambda^{(2)}(v) > 0 $ and $2\Lambda(v)^{-3}\Lambda^{(1)}(v)^2 - \Lambda(v)^{-2}\Lambda^{(2)}(v) > 0$. Thus $\cL_{nj}(\gamma_{1},\gamma_{2})$ is convex with respect to $\gamma_{2}$. For any given $\gamma_{2}$, the Hessian of $\cL_{nj}(\gamma_{1},\gamma_{2})$ with respect to $\gamma_{1}$ is 
			\[
			\frac{1}{n\kappa_{n}}\sum_{i=1}^{n}R_{1i}p(\bV_{i})p(\bV_{i})^{\T}\frac{2}{\Lambda(\gamma_{2}^{\T}p(\bV_{i}))},
			\]
			which is also positive semi-definite. This completes the proof.
		\end{proof}
		Proposition \ref{prop: optimization MVR} shows $\cL_{nj}(\gamma_{1},\gamma_{2})$ is block-wise convex with respect to $\gamma_{1}$ and $\gamma_{2}$ for $j = 1,\dots, d$ and hence the optimization problem \eqref{eq: MVR} can be solved efficiently by routine optimization algorithms. So far we have defined a nonparametric estimator for $\Pi_{j}(\bV)$, $\sigma_{j}(\bV)$ based on the sieve method \citep{chen2007large}. For our numerical experiments, the estimators $\widetilde{\Pi}_{j}$ and $\tilde{\sigma}_{j}$ are employed. In the simulation, we introduce a small regularization  $0.1(d_{V} + 1)(\|\gamma_{1}\|^{2} + \|\gamma_{2}\|^{2})$ into the loss function \eqref{eq: MVR} to further enhance the stability of the solution, where $d_{V}$ is the dimension of the first-phase variable $\bV$. In the simulation studies and real data analysis, we normalize all first-phase variables to the range $[0, 1]$ using min-max normalization and then use a second-order polynomial with full interactions between variables as the basis functions. This simple choice performs well in our numerical experiments, and we recommend it for practical use. However, we do not claim that this choice of basis functions is optimal in any sense. Identifying the optimal basis functions for the problem considered here remains an interesting direction for future research.
		
		The idea to recover the mean and variance simultaneously also appears in parametric heterogeneous regression literature and is shown to perform well in finite sample \citep{daye2012high,spady2018simultaneous}. Our estimator is an extension of the idea to the nonparametric literature. The proposed method has several nice properties compared to other nonparametric conditional mean and variance estimators, such as kernel smoothing, sieve least squares, and kernel ridge regression. First, it is computationally efficient as the conditional mean and variance can be estimated simultaneously by solving the optimization problem \eqref{eq: MVR}. Second, heteroscedasticity is considered in fitting the conditional mean model, which can deliver efficiency gains when fitting a model with many parameters and limited observations \citep{daye2012high}. Third, the conditional variance estimator is always positive. This is also a desirable property \citep{yu2004likelihood} which is not possessed by some classic existing methods, e.g., local linear kernel smoothing, sieve least squares, and kernel ridge regression. 
		
		\subsection{Estimate the Parameters of Interest}\label{subsec: one-step}
		Define $R_{2i} = 0$ for subjects with $R_{1i} = 1$. With some abuse of notation, let $R_{i} = R_{1i} + R_{2i}$ be the overall sampling indicator for the second phase sampling.
		Under the sampling procedure proposed in Section \ref{sec: est}, $\bV_{i}$ is measured for all the $n$ subjects, and $\bU_{i}$ is measured for subjects with $R_{i} = 1$. If a nonrandom sampling rule $\rho(\cdot)$ is used to select the subsequent sample in the second phase, we have the inclusion probability $P(R_{i} = 1\mid \bV_{i}) = \kappa_{n} + (1 - \kappa_{n})\rho(\bV_{i})$ for $i = 1,\dots, n$. Denote the sampling rule adopted in the second phase by $\tilde{\rho}(\cdot)$, where $\tilde{\rho}(\cdot)$ depends on the pilot sample and hence is random. However, it converges to some nonrandom sampling rule according to Theorem \ref{thm: convergence rule mv}. Thus the inclusion probability can be approximated by $\rho_{n}(\bV_{i}) = \kappa_{n} + (1 - \kappa_{n})\tilde{\rho}(\bV_{i})$. Let the inverse probability weighted estimator $\hat{\theta}_{\rm ipw}$ be the solution of the estimating equation
		\[\sum_{i=1}^{n}\frac{R_{i}\psi(\bV_{i}, \bU_{i}; \theta, \tilde{\eta})}{\rho_{n}(\bV_{i})} = 0.\] 
		The inverse probability weighted estimator $\hat{\theta}_{\rm ipw}$ is $\sqrt{n}$-consistent under certain regularity conditions but may not be efficient \citep{tsiatis2007semiparametric}. Based on $\hat{\theta}_{\rm ipw}$, an efficient estimator can be obtained through one-step estimation \citep{bickel1982adaptive}. Let $\widetilde{\Pi}(\cdot)$ be an estimate of $\Pi(\cdot)$ based on the pilot sample. According to the efficient influence function given in Lemma \ref{lem: EIF}, the one-step estimator is defined by
		\[
		\begin{aligned}
			\hat{\theta} & = \hat{\theta}_{\rm ipw} + \sum_{i=1}^{n}\frac{R_{i}\psi(\bV_{i}, \bU_{i}; \hat{\theta}_{\rm ipw}, \tilde{\eta})}{\rho_{n}(\bV_{i})} - \sum_{i=1}^{n}\left(\frac{R_{i}}{\rho_{n}(\bV_{i})} - 1\right)\widetilde{\Pi}(\bV_{i}) \\
			& = \hat{\theta}_{\rm ipw} - \sum_{i=1}^{n}\left(\frac{R_{i}}{\rho_{n}(\bV_{i})} - 1\right)\widetilde{\Pi}(\bV_{i}),
		\end{aligned}
		\] 
		The one-step estimator $\hat{\theta}$ is asymptotically normal and efficient under appropriate empirical process conditions \citep{van2012unified,wellner1996weak}.
		
		As noted by an anonymous reviewer, the use of the pilot sample breaks the i.i.d. structure across observations because the distribution of the sampling indicator in the second phase depend on the pilot sample. This complicates the theoretical analysis of the proposed one-step estimator $\hat{\theta}$. However, we note that, conditional on the pilot sample, the remaining observations are i.i.d., and the pilot sample typically constitutes only a small fraction of the total dataset. Intuitively, this suggests that the overall data distribution does not deviate substantially from the i.i.d. setting. Although not able to rigorously prove, we conjecture that this deviation does not cause serious issues. Our numerical results indicate that, despite of the non-i.i.d. structure, the proposed estimator $\hat{\theta}$ performs well in practice, suggesting that the non-i.i.d. structure and potential overfitting are not major concerns---at least in our numerical experiments. The simulations in Section~\ref{subsec: compare alt} show that $\hat{\theta}$ has negligible bias and achieves higher finite-sample efficiency than several alternative estimators. Based on these observations, we recommend the one-step estimator $\hat{\theta}$ for practical use, while leaving a rigorous investigation of its theoretical properties to future research. 
		
		\section{Additional Simulations}\label{sec: further simulation}
		\subsection{Response Mean}\label{subsec: response mean}
		
		In this section, we consider the response mean estimation problem where covariates are inexpensive and the response is hard to obtain. As in the main text, we set $n = 5000$. Let $\bZ$ be a $q$-dimensional covariate vector  with independent $U[-2.5,2.5]$ components, where $U[-2.5, 2.5]$ is the uniform distribution on $[-2.5, 2.5]$. Suppose the response
		\[Y = \theta_{0} + \bzeta_{q}^{\T}\bZ + \nu_{1}(\bZ) \epsilon, \]
		where $\theta_{0} = 1$, $\bzeta_{q} = (0.5/\sqrt{q}, \dots, 0.5/\sqrt{q})^{\T}$, $\nu_{1}(z) = \sqrt{0.1 + (2\bzeta_{q}^{\T}z)^{4}}$, $\epsilon$ is the standard normal error. In this example, we let $\bV = \bZ$ be the vector of first-phase variables and $\bU = Y$ be the second-phase variable. The parameter of interest is the response mean $\theta_{0} = E[Y]$. We take $\kappa_{n} = \varpi/\{1 + \log(\varpi n/q)\}$ in this simulation.
		Figure~\ref{fig: mean boxplot} is the boxplot based on the results of $500$ simulations.
		\begin{figure}
			\centering
			\subcaptionbox{}{\includegraphics[scale=0.35]{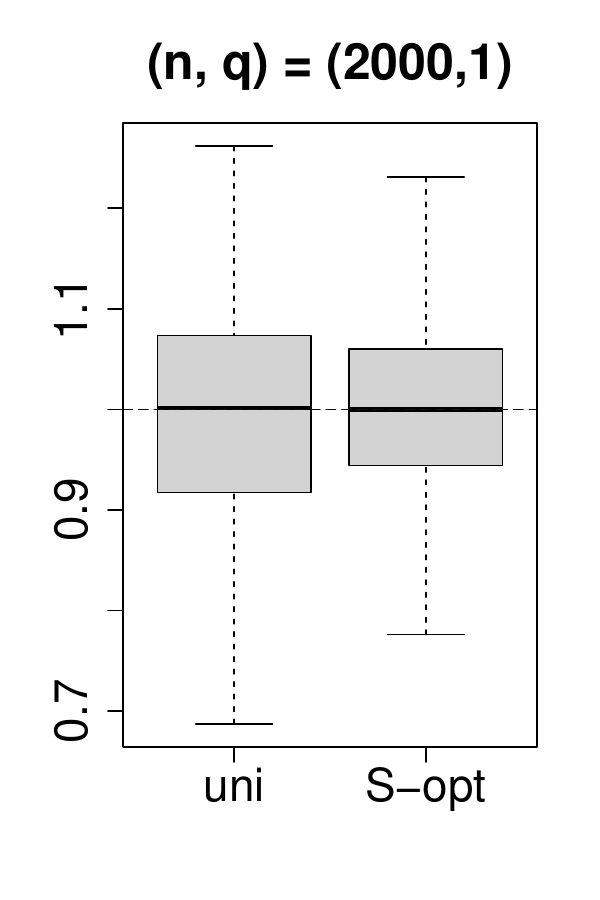}}
			\subcaptionbox{}{\includegraphics[scale=0.35]{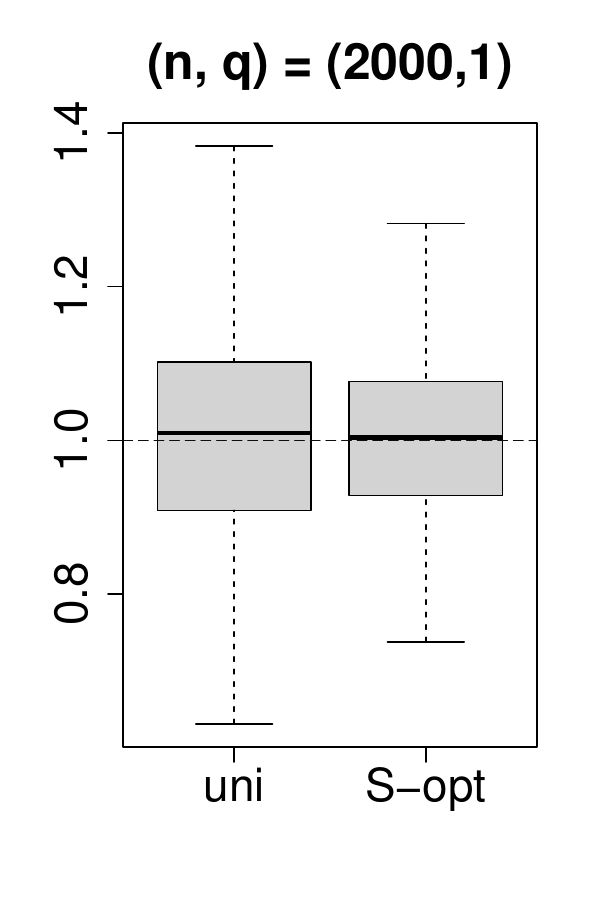}}
			\subcaptionbox{}{\includegraphics[scale=0.35]{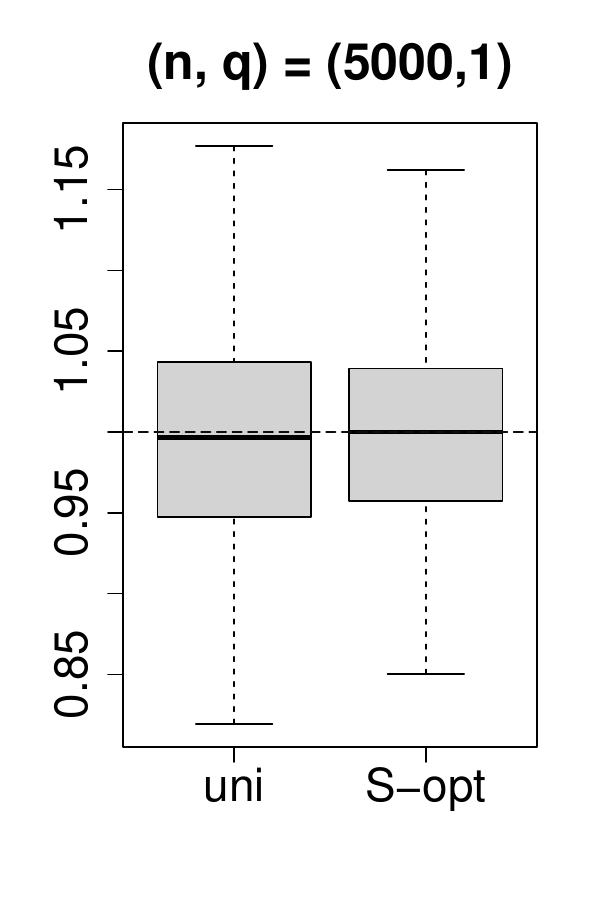}}
			\subcaptionbox{}{\includegraphics[scale=0.35]{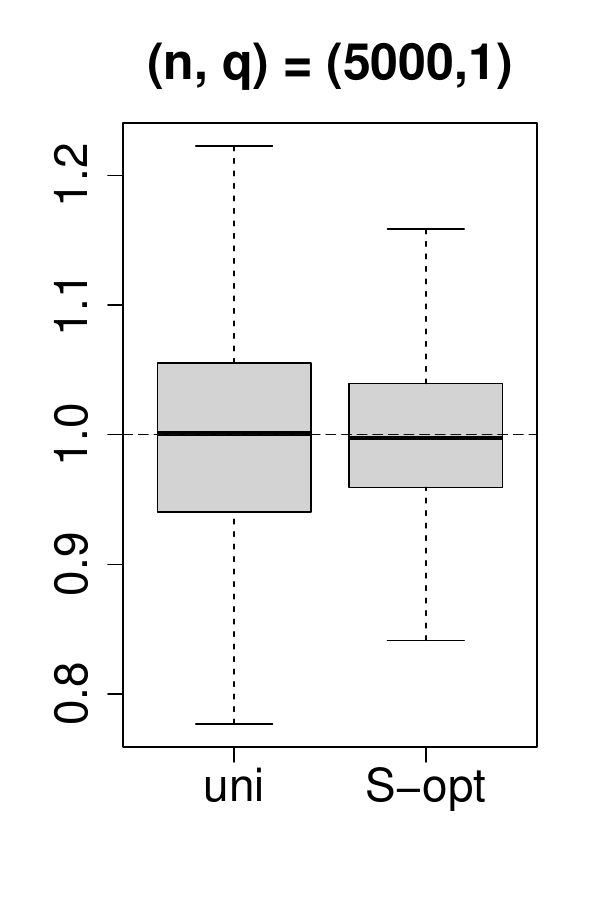}}
			\caption{Boxplots for the mean estimation with different combinations of $n$ and $q$; dashed lines are the true values.}\label{fig: mean boxplot}
		\end{figure}
		As can be seen in Fig.~\ref{fig: mean boxplot}, the estimation efficiency is improved under $\tilde{\rho}_{\cS}$ compared to the uniform rule. The improvement is observed for different combinations of $n$ and $q$, and is particularly pronounced when $q = 5$.
		The REs are $1.5038$, $1.7409$, $1.5103$, and $2.1777$ when $(n,q) = (2000, 1), (2000, 5), (5000, 1)$, and $(5000, 5)$, which indicates that the efficiency is significantly improved under the proposed optimal sampling rule compared to that under the uniform rule.
		
		In the following, we consider the problem of multi-dimensional parameter estimation and evaluate the effectiveness of the sampling rule $\tilde{\rho}_{\cC}$ and $\tilde{\rho}_{\cG}$. We consider a two-dimensional response variable $Y = (Y_{1}, Y_{2})^{\T}$. Suppose
		\[
		\begin{aligned}
			&Y_{1} = \theta_{01} - \bzeta_{q}^{\T}\bZ + \nu_{1}(\bZ)\epsilon_{1},\\
			&Y_{2} = \theta_{02} + \sin(\bzeta_{q}^{\T}\bZ) + \nu_{2}(\bZ)\epsilon_{2},
		\end{aligned}
		\]
		where $\theta_{01} = 1$, $\theta_{02} = 0$, $\nu_{1}(z)$ is introduced in the scalar case, $\nu_{2}(z) = \exp(2\bzeta_{q}^{\T}z)$, $\epsilon_{1}$ and $\epsilon_{2}$ are independent standard normal errors, and $\bZ$ is defined in the same way as in the scalar case. We consider the estimation of the two-dimensional response mean $\theta_{0} = (\theta_{01}, \theta_{02})^{\T} = (E[Y_{1}], E[Y_{2}])^{\T}$ under two-phase designs.
		Table~\ref{table: mv mean} reports the bias and standard error (SE) of the one-step estimator with different sampling rules based on $500$ simulations, and the RE compared to the uniform rule. 
		\begin{table}
			\centering
			\caption{Bias, SE, and RE in two-dimensional mean estimation}
			\begin{tabular}{llcccccc}
				\multirow{2}{*}{$(n, q)$}&\multirow{2}{*}{Rule} & \multicolumn{3}{c}{Estimate of $\theta_{01}$}  & \multicolumn{3}{c}{Estimate of $\theta_{02}$}\\
				& & Bias & SE & RE & Bias & SE& RE \\
				\specialrule{0em}{-4pt}{-4pt}\\
				\multirow{3}{*}{$(2000, 1)$}
				&uni & 0.0017 & 0.1224 & 1.0000 & -0.0012 & 0.1601 & 1.0000\\
				&C-opt & -0.0028 & 0.0897 & 1.8620 & -0.0001 & 0.1201 & 1.777\\
				&G-opt & 0.0007 & 0.0965 & 1.6088 & 0.0050 & 0.1240 & 1.667\\
				\specialrule{0em}{-4pt}{-4pt}\\
				\multirow{3}{*}{$(2000, 5)$}
				&uni & 0.0068 & 0.1441 & 1.0000 & 0.0254 & 0.2305 & 1.0000\\
				&C-opt & 0.0024 & 0.1031 & 1.9535 & 0.0053 & 0.1526 & 2.2816\\
				&G-opt & 0.0060 & 0.1068 & 1.8205 & 0.0036 & 0.1522 & 2.2936\\
				\specialrule{0em}{-4pt}{-4pt}\\
				\multirow{3}{*}{$(5000, 1)$}
				&uni & 0.0008 & 0.0719 & 1.0000 & 0.0047 & 0.0981 & 1.0000\\
				&C-opt & -0.0038 & 0.0601 & 1.4312 & 0.0019 & 0.0754 & 1.6928\\
				&G-opt & 0.0010 & 0.0599 & 1.4408 & -0.0003 & 0.0722 & 1.8461\\
				\specialrule{0em}{-4pt}{-4pt}\\			
				\multirow{3}{*}{$(5000, 5)$}
				&uni & -0.0056 & 0.0889 & 1.0000 & -0.0047 & 0.1639 & 1.0000\\
				&C-opt & -0.0006 & 0.0665 & 1.7871 & -0.0018 & 0.0947 & 2.9954\\
				&G-opt & 0.0010 & 0.0620 & 2.0560 & -0.0012 & 0.0962 & 2.9027\\
			\end{tabular}\label{table: mv mean}
		\end{table}
		In all cases, the SEs under the sampling rules $\tilde{\rho}_{\cC}$ and $\tilde{\rho}_{\cG}$ are smaller than that under the uniform rule. In some cases, the improvement is very significant with a RE close to $3$. 
		
		\subsection{Linear Regression Coefficient}
		In this subsection, we consider the problem of estimating linear regression coefficients when response variables and a part of covariates are measured in the first phase and other covariates are measured in the second phase. Define $\bZ$ in the same way as in Section \ref{subsec: response mean}. Suppose
		\[X = \sin(\bzeta_{q}^{\T}\bZ) + \nu_{2}(\bZ)\epsilon_{x}, \]
		\[Y = \bzeta_{q}^{\T}\bZ + \theta_{0}X + \epsilon_{y},\] 
		where $\theta_{0} = 1$, $\nu_{2}(z)$ is defined in the last Section \ref{subsec: response mean}, and $\epsilon_{x}$ and $\epsilon_{y}$ are independent and follow a standard normal distribution. Let $V = (Y, Z^{\T})^{\T}$ be the first-phase variable vector and $U = X$ be the second-phase variable. The estimation of the regression coefficient $\theta_{0}$ is considered in this simulation. We take $\kappa_{n} = \varpi/[1 + \log\{\varpi n/(q + 1)\}]$ in this simulation.
		Figure~\ref{fig: reg boxplot} contains boxplots of the estimator in $500$ simulations across different combinations of $n$ and $q$.
		\begin{figure}
			\centering
			\subcaptionbox{}{\includegraphics[scale=0.35]{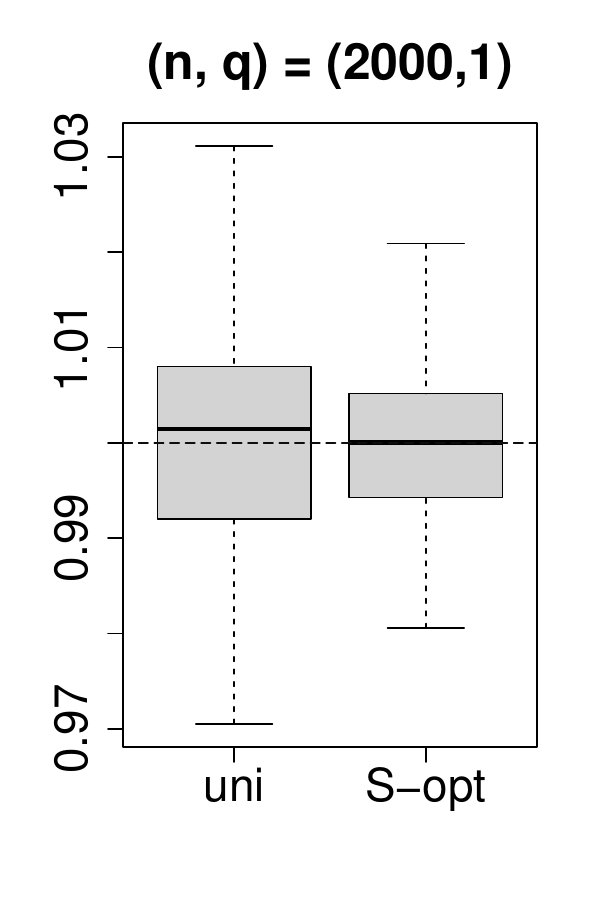}}
			\subcaptionbox{}{\includegraphics[scale=0.35]{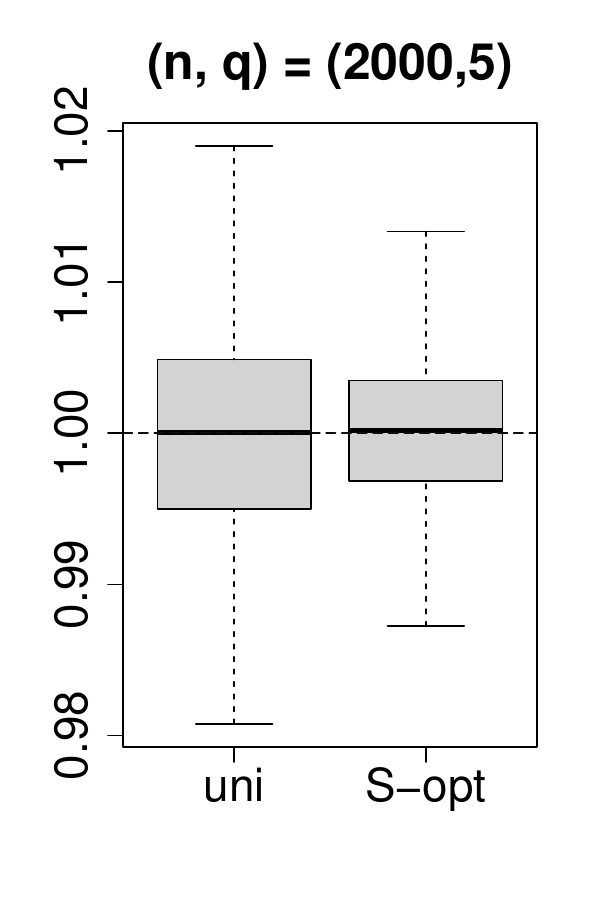}}
			\subcaptionbox{}{\includegraphics[scale=0.35]{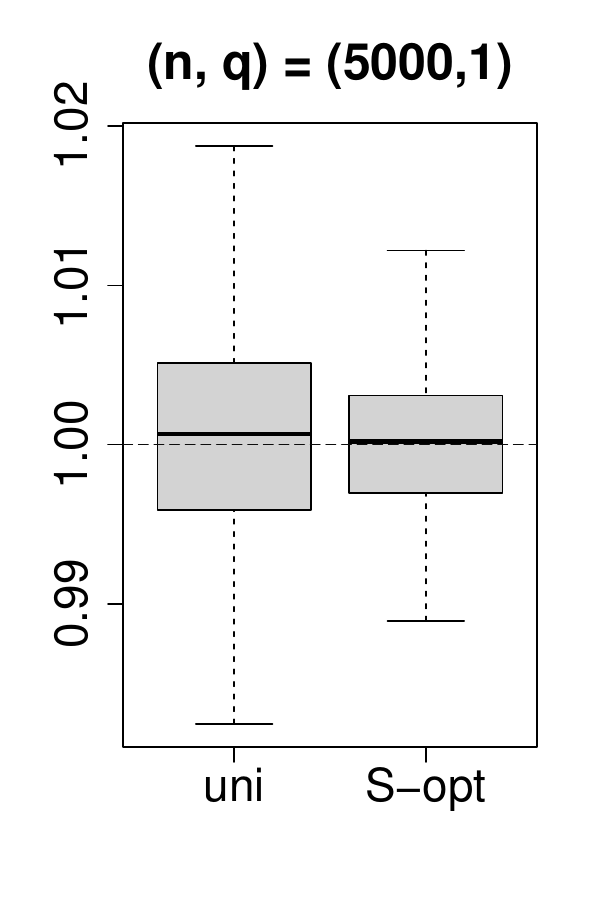}}
			\subcaptionbox{}{\includegraphics[scale=0.35]{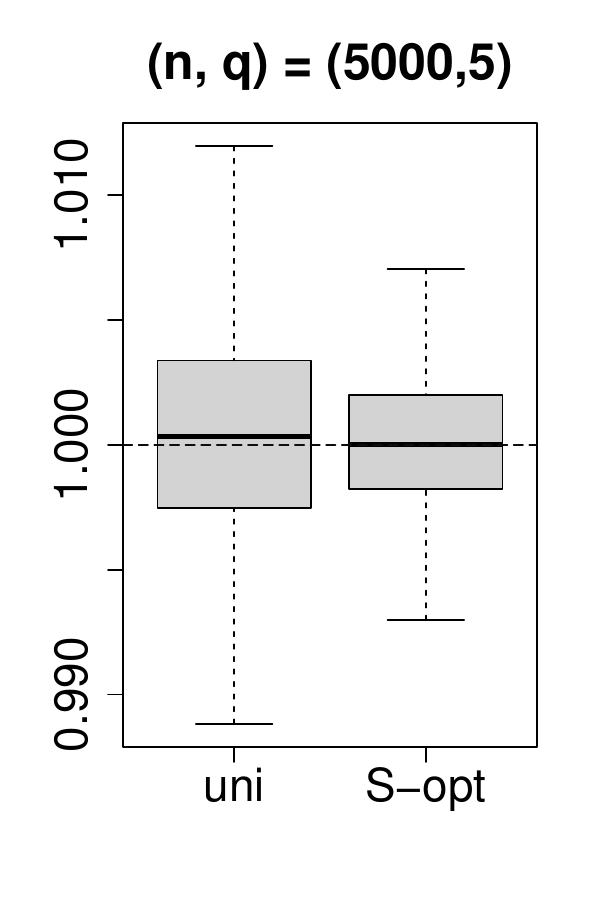}}
			\caption{Boxplots for linear regression coefficient estimation with different combinations of $n$ and $q$; dashed lines are the true values.}\label{fig: reg boxplot}
		\end{figure}
		As can be seen in Fig.~\ref{fig: reg boxplot}, the accuracy of the estimator is improved under $\tilde{\rho}_{\cS}$ compared to that under the uniform rule. The improvements are observed across different combinations of $n$ and $q$, and is particularly pronounced when $(n, q) = (5000, 5)$. The REs are $2.1175$, $2.4225$, $2.3096$, and $2.3206$ when $(n,q) = (2000, 1), (2000, 5), (5000, 1)$, and $(5000, 5)$. 
		
		Next, we consider the case with a two-dimensional regression coefficient of interest. The covariate vector $\bZ$ is defined in the same way as in Section \ref{subsec: response mean}. Let $\bX = (X_{1}, X_{2})^{\T}$  be a two-dimensional covariate vector which satisfies
		\[X_{1} = \bzeta_{q}^{\T}\bZ + \nu_{1}(\bZ)\epsilon_{x1}, \]
		\[X_{2} = - \bzeta_{q}^{\T}\bZ + \nu_{2}(\bZ)\epsilon_{x2}, \]
		where $\nu_{1}(z)$ and $\nu_{2}(z)$ are defined in Section \ref{subsec: response mean}, $\epsilon_{x1}$, $ \epsilon_{x2}$ are independent standard normal variables. The response variable $Y$ satisfies
		\[Y = \bzeta_{q}^{\T}\bZ + \theta_{0}^{\T}\bX + \epsilon_{y},\]
		where $\theta_{0}^{T}=(\theta_{01}, \theta_{02}) = (0, 1)$ and $\epsilon_{y}$ follows a standard normal distribution.
		
		Table~\ref{table: mv reg} reports the bias, SE of the one-step estimator, and the RE compared to the uniform rule based on $500$ simulations.
		\begin{table}
			\centering
			\caption{Bias, SE, and RE in two-dimension regression coefficient estimation}
			\begin{tabular}{llcccccc}
				\multirow{2}{*}{$(n, q)$}&\multirow{2}{*}{Rule} & \multicolumn{3}{c}{Estimate of $\theta_{01}$}  & \multicolumn{3}{c}{Estimate of $\theta_{02}$}\\
				& & Bias & SE & RE & Bias & SE& RE \\
				\specialrule{0em}{-4pt}{-4pt}\\
				\multirow{3}{*}{$(2000, 1)$}
				&uni & -0.0005 & 0.0145 & 1.0000 & 0.0010 & 0.0110 & 1.0000\\
				&C-opt & -0.0012 & 0.0126 & 1.3243 & 4e-04 & 0.0092 & 1.4296\\
				&G-opt &-0.0003 & 0.0133 & 1.1886 & 8e-04 & 0.0088 & 1.5625\\
				\specialrule{0em}{-4pt}{-4pt}\\
				\multirow{3}{*}{$(2000, 5)$}
				&uni & 0.0000 & 0.0126 & 1.0000 & 0.0003 & 0.0076 & 1.0000\\
				&C-opt & 0.0005 & 0.0105 & 1.4400 & 0.0004 & 0.0057 & 1.7778\\
				&G-opt & 0.0000 & 0.0103 & 1.4965 & 0.0003 & 0.0057 & 1.7778\\
				\specialrule{0em}{-4pt}{-4pt}\\
				\multirow{3}{*}{$(5000, 1)$}
				&uni & 0.0001 & 0.0094 & 1.0000 & 0.0009 & 0.0071 & 1.0000\\
				&C-opt & -0.0003 & 0.0079 & 1.4158 & 0.0005 & 0.0055 & 1.6664\\
				&G-opt & -0.0003 & 0.0078 & 1.4523 & 0.0003 & 0.0054 & 1.7287\\
				\specialrule{0em}{-4pt}{-4pt}\\
				\multirow{3}{*}{$(5000, 5)$}
				&uni & 0.0008 & 0.0075 & 1.0000 & 0.0003 & 0.0047 & 1.0000\\
				&C-opt & 0.0001 & 0.0062 & 1.4633 & 0.0001 & 0.0032 & 2.1572\\
				&G-opt & 0.0003 & 0.0060 & 1.5625 & 0.0001 & 0.0032 & 2.1572\\
			\end{tabular}\label{table: mv reg}
		\end{table}
		As seen in Table~\ref{table: mv reg}, the SEs under under the sampling rules $\tilde{\rho}_{\cC}$ and $\tilde{\rho}_{\cG}$ are smaller than those under the uniform rule in all cases. The REs are larger than two in some cases. 
		
		\subsection{Sampling Design with Different Priorities under the Multi-dimensional Setting in Section \ref{sec: sim}}\label{subsec: sim prioritize}
		In this section, under the multi-dimensional setting in Section \ref{sec: sim}, we illustrate the numerical effect of the priority parameter $a = (a_{1}, 1 - a_{1})$ discussed in Remark \ref{remark: priority}. We set $(n, q) = (5000, 1)$. Table \ref{table: mv ATE priortize} presents the bias, SE of the estimator, and the RE compared to the uniform rule, based on $500$ simulations. We consider the one-step estimator under the estimated optimal rule for $\theta_{01}$, $\theta_{02}$, and the G-opt rule with $a_{1} = 0.05, 0.5$ and $0.95$. The results under the uniform sampling rule are also reported for reference. Table \ref{table: mv ATE priortize} shows that biases of the estimator are small under all sampling rules. Notably, the SE for estimating $\theta_{01}$ decreases as $a_{1}$ increases, accompanied by a modest increase in the SE for estimating $\theta_{02}$. Notably, the SEs for both parameter components are smaller than those obtained under the uniform sampling rule, no matter what value $a_{1}$ takes. This desirable property is not achieved for the estimated optimal rule for $\theta_{01}$ or $\theta_{02}$.	
		\begin{table}
			\centering 
			\caption{Bias, SE, and RE in two-dimensional average treatment effect estimation with $(n, q) = (5000, 1)$}
			\begin{tabular}{lcccccc}
				\multirow{2}{*}{Rule} & \multicolumn{3}{c}{Estimate of $\theta_{01}$}  & \multicolumn{3}{c}{Estimate of $\theta_{02}$}\\
				& Bias & SE & RE & Bias & SE& RE \\
				\specialrule{0em}{-4pt}{-4pt}\\
				
				uni & 0.0273 & 0.3113 & 1.0000 & -0.0101 & 0.2787 & 1.0000\\
				S-opt ($\theta_{01}$) & 0.0206 & 0.2463 & 1.5975 & -0.0024 & 0.2977 & 0.8764\\
				S-opt ($\theta_{02}$) & 0.0516 & 0.3639 & 0.7318 & 0.0009 & 0.1851 & 2.2671\\
				G-opt ($a_{1} = 0.05$) & 0.0427 & 0.2814 & 1.2238 & 0.0026 & 0.1982 & 1.9773\\
				G-opt ($a_{1} = 0.5$) & 0.0254 & 0.2585 & 1.4502 & 0.0058 & 0.2222 & 1.5732\\
				G-opt ($a_{1} = 0.95$) & 0.0262 & 0.2517 & 1.5296 & -0.0053 & 0.2262 & 1.5181\\
			\end{tabular}\label{table: mv ATE priortize}
		\end{table}
		
		\subsection{Comparison with Alternative Estimators}\label{subsec: compare alt}
		
		In this section, we further investigate the numerical performance of the proposed one-step estimator $\hat{\theta}$ and compare it with several alternative estimators.
		There are multiple alternative ways to construct an efficient estimator in two-phase studies. For example, one may construct an estimator, based on the efficient influence function and the one-step estimation technique, using all observations except those in the pilot sample. Denote the resulting estimator by $\hat{\theta}_{\rm ex}$. Conditional on the pilot sample, the remaining observations are i.i.d. Thus, the asymptotic properties of $\hat{\theta}_{\rm ex}$ can be established using standard arguments conditional on the pilot sample. In addition, the efficiency loss lead by excluding the pilot sample is asymptotically negligible because $\kappa_{n} \to 0$. Another reasonable approach is to perform an inverse-variance-weighted meta-analysis combining the pilot sample estimator $\tilde{\theta}$ and the estimator $\hat{\theta}_{\rm ex}$. Denote the resulting estimator by $\hat{\theta}_{\rm ivw}$. Figure \ref{fig: ATE alt boxplot} presents boxplots of different average treatment effect estimators under the optimal sampling rule for a scalar parameter $\tilde{\rho}_{\cS}$ over $500$ simulations with $q = 1$ and $n = 2000, 5000, 20000, 50000$.  For reference, the results for the one-step estimator under the uniform sampling rule are also included.
		\begin{figure}
			\centering
			\subcaptionbox{}{\includegraphics[scale=0.35]{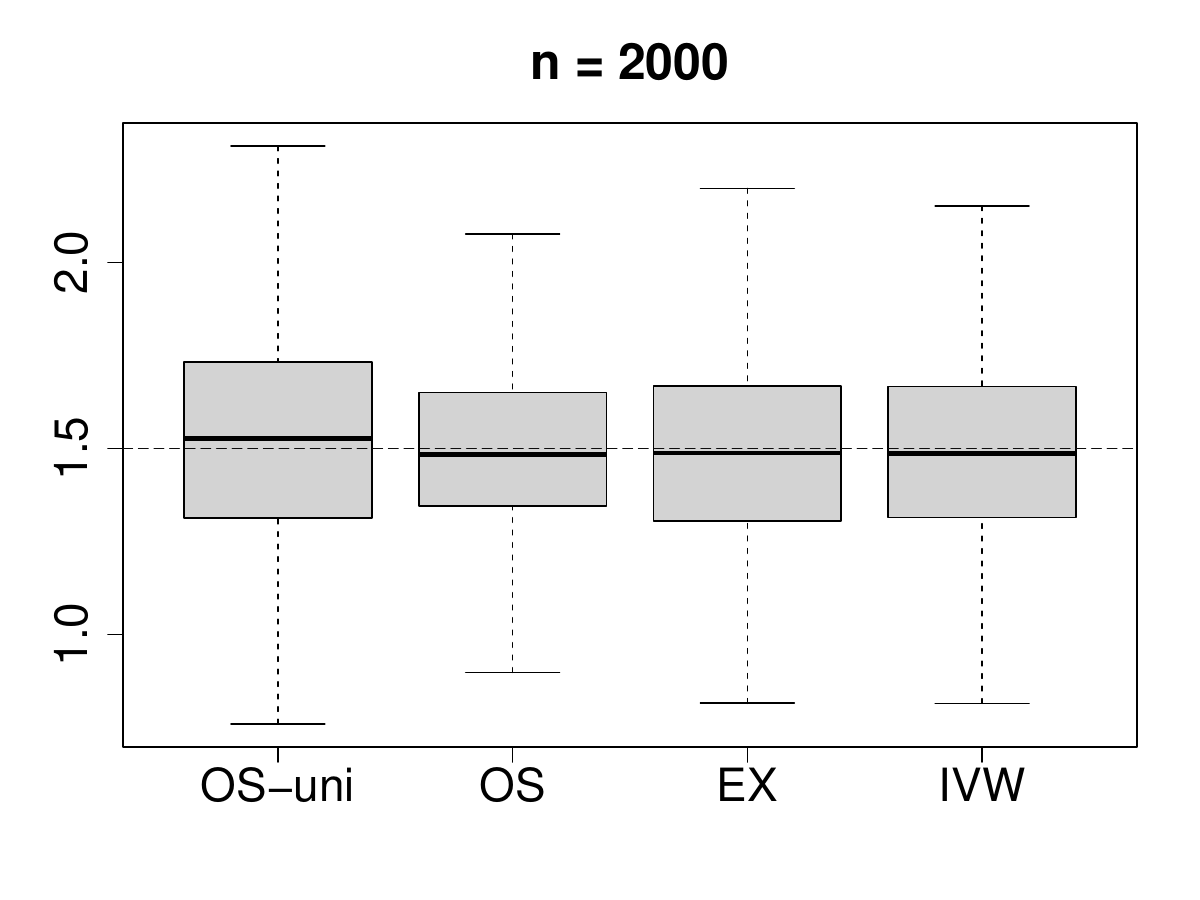}}
			\subcaptionbox{}{\includegraphics[scale=0.35]{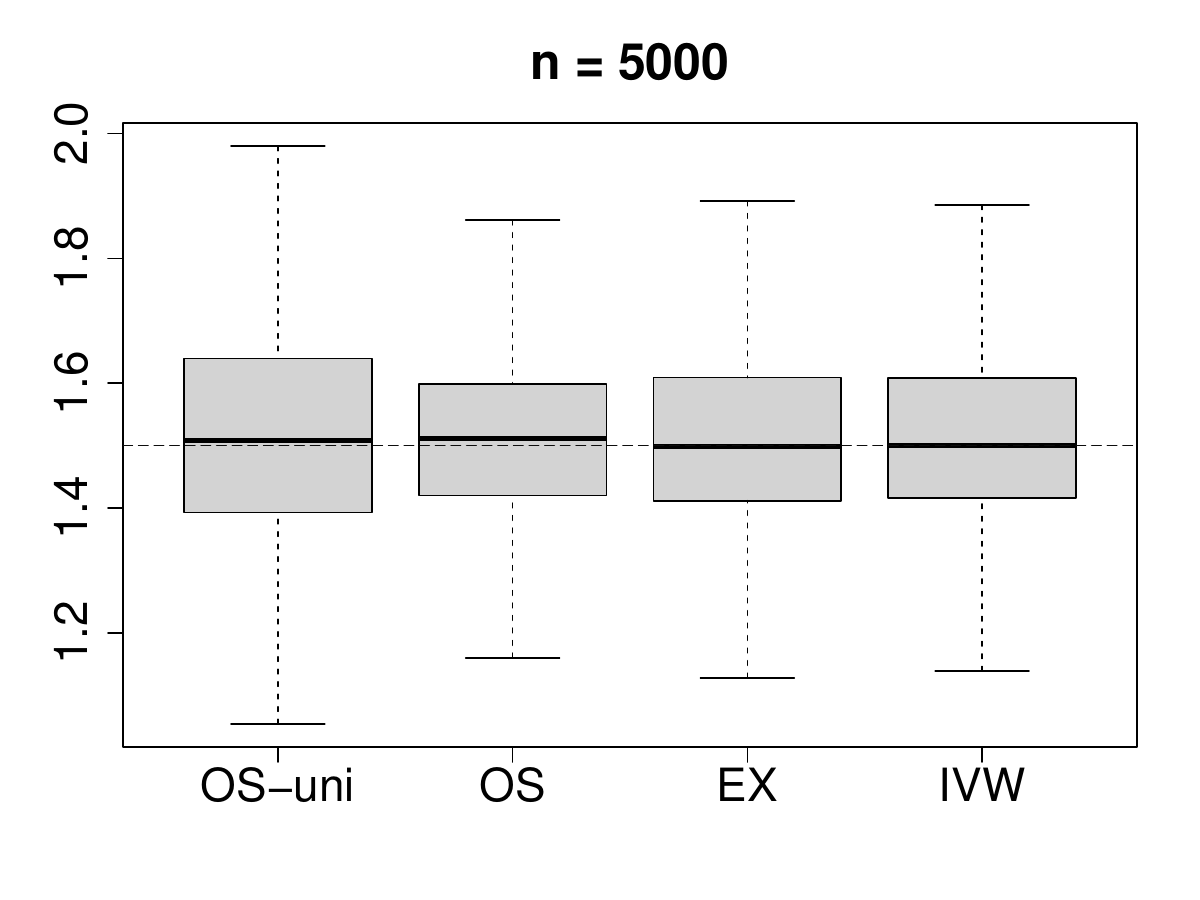}}
			\subcaptionbox{}{\includegraphics[scale=0.35]{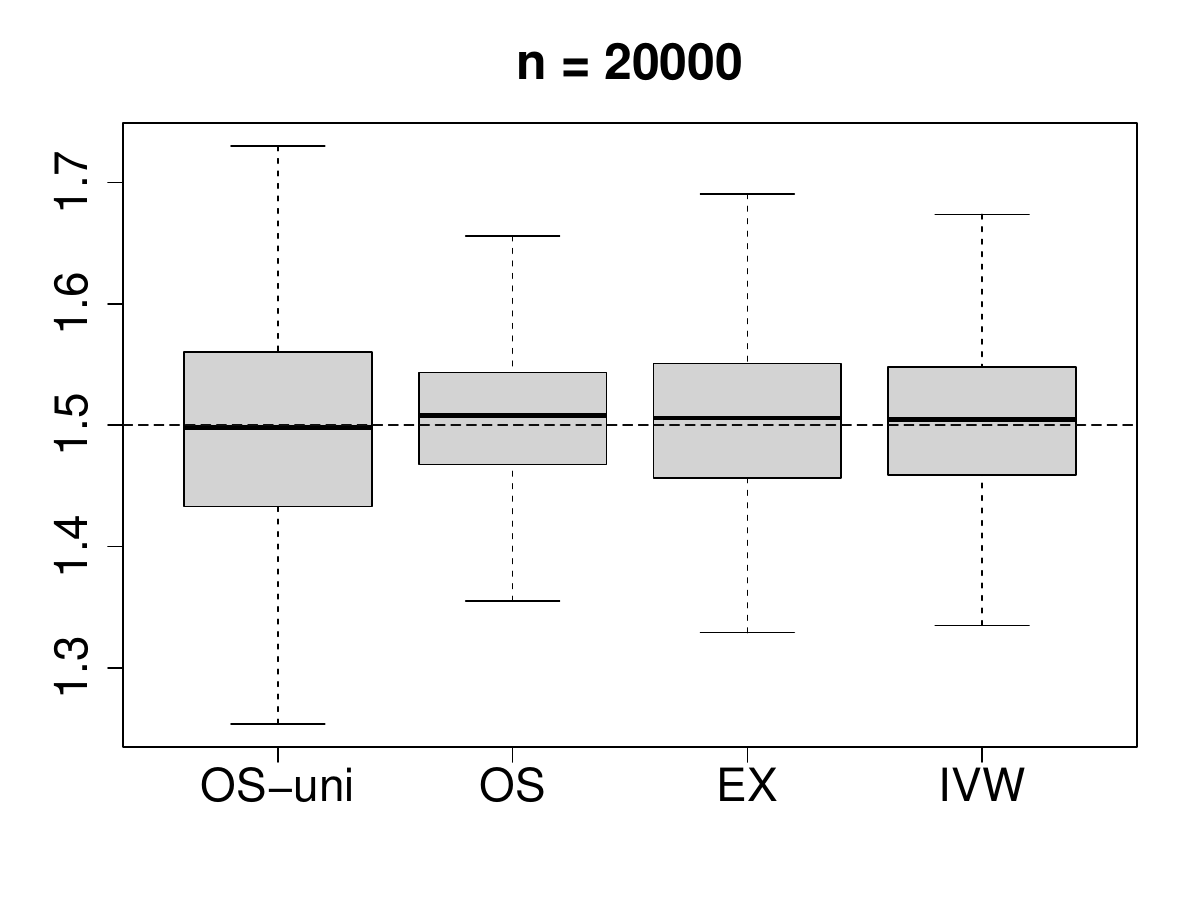}}
			\subcaptionbox{}{\includegraphics[scale=0.35]{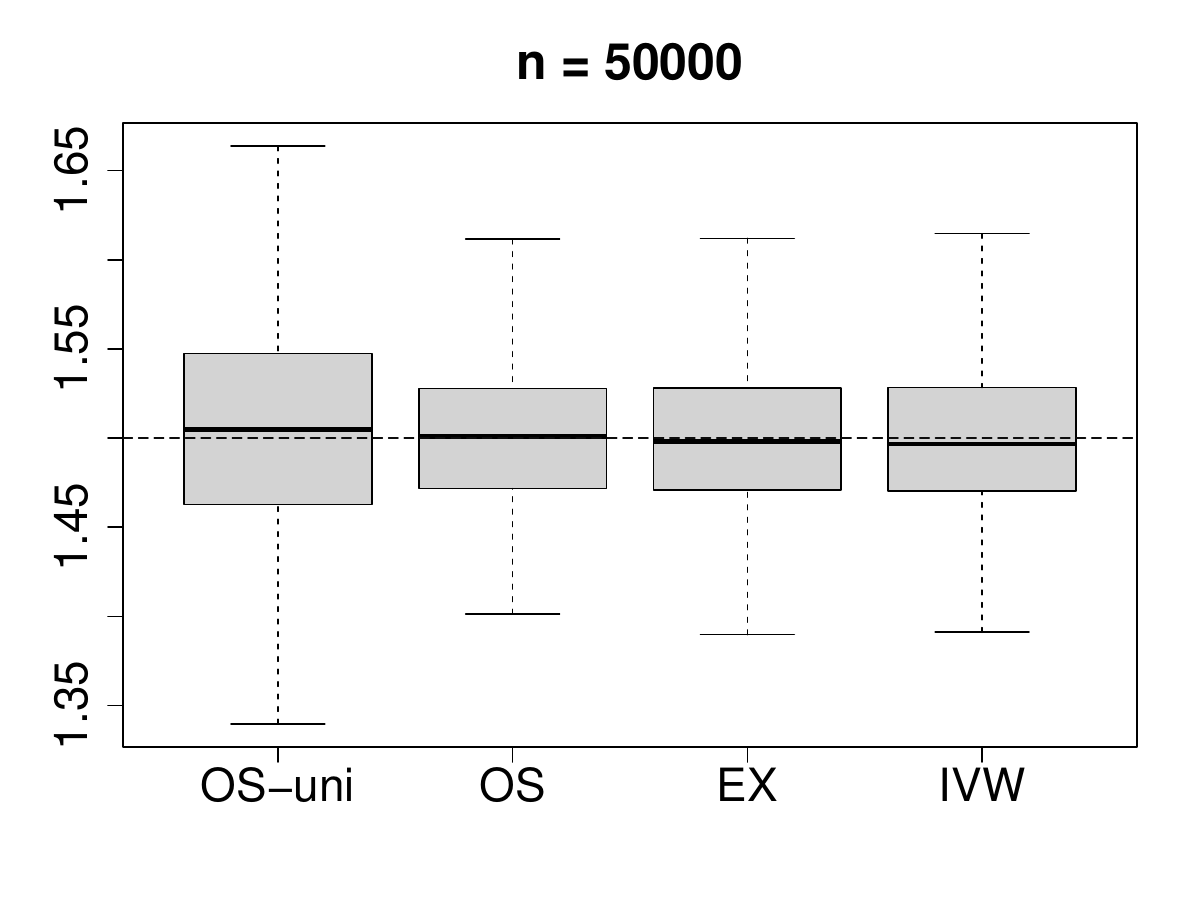}}
			\caption{Boxplots for the average treatment effect estimation under $q = 1$ and $n = 2000, 5000, 20000, 50000$; OS-uni: $\hat{\theta}$ under the uniform sampling rule; OS: $\hat{\theta}$ under $\tilde{\rho}_{\cS}$; EX: $\hat{\theta}_{\rm ex}$ under $\tilde{\rho}_{\cS}$; IVW: $\hat{\theta}_{\rm ivw}$ under $\tilde{\rho}_{\cS}$; dashed lines are the true values.}\label{fig: ATE alt boxplot}
		\end{figure}
		Figure \ref{fig: ATE alt boxplot} shows that, under $\tilde{\rho}_{\cS}$,  all three estimators have higher efficiency than the one-step estimator under the uniform sampling rule.  The performance of the three estimators under $\tilde{\rho}_{\cS}$ is similar when $n = 50000$. However, for smaller sample sizes ($n = 2000, 5000, 50000$), the proposed one-step estimator $\hat{\theta}$ demonstrates better finite-sample efficiency than both $\hat{\theta}_{\rm ex}$ and $\hat{\theta}_{\rm ivw}$ under $\tilde{\rho}_{\cS}$. In addition, the confidence interval based on $\hat{\theta}$ and normal approximation performs well in the simulation.
		Coverage rates of the confidence intervals constructed using $\hat{\theta}$ and the influence function-based standard error estimator are $94.4\%$, $93.8\%$, $96.8\%$, and $96.8\%$ when $n = 2000, 5000, 20000, 50000$, respectively.
		\bibliographystyle{asa}
		\bibliography{TwoPhase-arXiv.bib}
\end{document}